\documentclass[format=acmsmall, review=false, screen=true]{acmart_edit}

\usepackage{booktabs} % For formal tables

\acmJournal{CSUR}
\acmVolume{v}
\acmNumber{n}
\acmArticle{0}
\acmYear{2019}
\acmMonth{4}
\copyrightyear{2019}
\setcopyright{acmlicensed}
\usepackage{latexsym,amsmath,amssymb,color}
\usepackage[shortlabels]{enumitem}
\usepackage{xspace}
\usepackage{wrapfig}
\setlist{nolistsep}

\newenvironment{myquote}[1][\parindent]%
  {\list{}{\leftmargin=#1\rightmargin=#1}\item[]\it }%
  {\endlist}

\DeclareSymbolFont{frenchscript}{OMS}{ztmcm}{m}{n}
\DeclareMathSymbol{\A}{\mathord}{frenchscript}{65}    % set of CCS names
\DeclareMathSymbol{\Ce}{\mathord}{frenchscript}{67}   % component set
\DeclareMathSymbol{\Ge}{\mathord}{frenchscript}{71}   % liveness property
\DeclareMathSymbol{\He}{\mathord}{frenchscript}{72}   % predicate on states
\newcommand{\I}{\mathcal{I}}                          % instruction set
\DeclareMathSymbol{\Lab}{\mathord}{frenchscript}{76}  % set of CCS labels
\DeclareMathSymbol{\Pow}{\mathord}{frenchscript}{80}  % powerset
\DeclareMathSymbol{\T}{\mathord}{frenchscript}{84}    % task set
\DeclareMathSymbol{\V}{\mathord}{frenchscript}{88}    % set of process variables (X)
\DeclareMathAlphabet{\mathbbm}{U}{bbm}{m}{n}          % blackboard bold
\newcommand{\IT}{\mathbbm{T}}               % set of valid expressions / terms
\newcommand{\IN}{\mathbbm{N}}               % natural numbers
\newcommand{\Sec}[1]{Section~\ref{sec:#1}}
\newcommand{\Sects}{Sections}
\newcommand{\App}[1]{Appendix~\ref{app:#1}}
\newcommand{\df}[1]{Definition~\ref{df:#1}}
\newcommand{\thm}[1]{Theorem~\ref{thm:#1}}
\newcommand{\pr}[1]{Proposition~\ref{pr:#1}}

\newcommand{\ex}[1]{Example~\ref{ex:#1}}
\newcommand{\tab}[1]{Table~\ref{tab:#1}}
\newcommand{\fig}[1]{Figure~\ref{#1}}
\newcommand{\exs}{Examples\xspace}
\makeatletter
\def\comesfrom{\@transition\leftarrowfill}
\def\goesto{\@transition\rightarrowfill}
\def\ngoesto{\@transition\nrightarrowfill}
\def\Goesto{\@transition\Rightarrowfill}
\def\nGoesto{\@transition\nRightarrowfill}
\def\xmapsto{\@transition\mapstofill}
\def\nxmapsto{\@transition\nmapstofill}
\def\@transition#1{\@@transition{#1}}
\newbox\@transbox
\newbox\@arrowbox
\newbox\@downbox
\def\@@transition#1#2%
   {\setbox\@transbox\hbox
      {\vrule height 1.5ex depth .8ex width 0ex\hskip0.25em$\scriptstyle#2$\hskip0.25em}
   \ifdim\wd\@transbox<1.5em
      \setbox\@transbox\hbox to 1.5em{\hfil\box\@transbox\hfil}\fi
   \setbox\@arrowbox\hbox to \wd\@transbox{#1}
   \ht\@arrowbox\z@\dp\@arrowbox\z@
   \setbox\@transbox\hbox{$\mathop{\box\@arrowbox}\limits^{\box\@transbox}$}
   \dp\@transbox\z@\ht\@transbox 10pt
   \mathrel{\box\@transbox}}
\def\nrightarrowfill{$\m@th\mathord-\mkern-6mu%
  \cleaders\hbox{$\mkern-2mu\mathord-\mkern-2mu$}\hfill
  \mkern-6mu\mathord\not\mkern-2mu\mathord\rightarrow$}
\def\Rightarrowfill{$\m@th\mathord=\mkern-6mu%
  \cleaders\hbox{$\mkern-2mu\mathord=\mkern-2mu$}\hfill
  \mkern-6mu\mathord\Rightarrow$}
\def\nRightarrowfill{$\m@th\mathord=\mkern-6mu%
  \cleaders\hbox{$\mkern-2mu\mathord=\mkern-2mu$}\hfill
  \mkern-6mu\mathord\not\mathord\Rightarrow$}
\def\mapstofill{$\m@th\mathord\mapstochar\mathord-\mkern-6mu%
  \cleaders\hbox{$\mkern-2mu\mathord-\mkern-2mu$}\hfill
  \mkern-6mu\mathord\rightarrow$}
\def\nmapstofill{$\m@th\mathord\mapstochar\mathord-\mkern-6mu%
  \cleaders\hbox{$\mkern-2mu\mathord-\mkern-2mu$}\hfill
  \mkern-6mu\mathord\not\mkern-2mu\mathord\rightarrow$}
\makeatother %%%%% end of arrows definition
\newcommand{\goto}[1]{\stackrel{#1}{\longrightarrow}} % transition
\newcommand{\Tr}{\textit{Tr}}                         % transitions
\newcommand{\source}{\textit{source\/}}               % source of transition
\newcommand{\target}{\textit{target\/}}               % target of transition
\newcommand{\instr}{\textit{instr\/}}                 % program instruction(s) involved in tr
\newcommand{\comp}{\textit{comp\/}}                   % components involved in transition
\newcommand{\cmp}{\textit{cp\/}}                      % components involved in instruction
\newcommand{\ac}{\textit{npc\/}}                      % components causing transition
\newcommand{\pc}{\textit{afc\/}}                      % components affected by transition
\newcommand{\defis}{\stackrel{{\it def}}{=}}
\newcommand{\X}{x}
\newcommand{\Y}{y}
\newcommand{\W}{{\rm W}}
\newcommand{\St}{{\rm S}}
\newcommand{\Ac}{{\rm A}}
\newcommand{\Ts}{{\rm T}}
\newcommand{\D}{{\rm D}}
\newcommand{\In}{{\rm I}}
\newcommand{\Sy}{{\rm Z}}
\newcommand{\Ex}{{\rm Ext}}
\renewcommand{\Pr}{{\rm Pr}}
\newcommand{\Fu}{{\rm Fu}}
\newcommand{\E}{{\rm E}}
\newcommand{\Cp}{{\rm C}}
\newcommand{\Gr}{{\rm G}}
\newcommand{\J}{{\rm J}}
\newcommand{\dom}{{\it dom}}                          % domain
\newcommand{\Left}{\textsc{l}}
\newcommand{\R}{\textsc{r}}
\newcommand{\M}{\textsc{m}}
\newcommand{\cT}{{\rm T}}                             % set of closed terms
\newcommand{\aconc}{\mathrel{\mbox{$\smile\hspace{-1.15ex}\raisebox{2.5pt}{$\scriptscriptstyle\bullet$}$}}}
\newcommand{\naconc}{\mathrel{\mbox{$\,\not\!\smile\hspace{-1.15ex}\raisebox{2.5pt}{$\scriptscriptstyle\bullet$}$}}}
\newcommand{\conc}{\smile}                            % concurrency relation
\newcommand{\nconc}{\,\not\!\smile}
\newcommand{\dcup}{\stackrel{\mbox{\huge .}}{\cup}}   % disjoint union
\newcommand{\plat}[1]{\raisebox{0pt}[0pt][0pt]{#1}}   % no vertical space
\newcommand{\init}[2]{{\bf initialise}\ \ensuremath{#1}{\ \bf to\ }\ensuremath{#2}}
\theoremstyle{acmdefinition}
\newtheorem{exam}{Example }

\begin{document}
\makeatletter
\@printpermissionfalse
\@printcopyrightfalse
\@acmownedfalse
 \acmPrice{}
 \def\@formatdoi#1{}
\makeatother
\title{Progress, Justness and Fairness}  

\author{Rob van Glabbeek}
\author{Peter H\"ofner}
%\orcid{1234-5678-9012-3456}
\affiliation{%
  \institution{Data61, CSIRO and
  %Computer Science and Engineering, 
  UNSW%
  \country{Australia}}
  }

\begin{abstract}
Fairness assumptions are a valuable tool when reasoning about systems. 
In this paper, we classify several fairness properties found in the literature and argue 
that most of them are too restrictive for many applications.
As an alternative we introduce the concept of \mbox{justness}.
\end{abstract}

% The code below has been generated by the tool at
% http://dl.acm.org/ccs.cfm
\begin{CCSXML}
<ccs2012>
	<concept>
		<concept_id>10003752.10003753.10003761.10003764</concept_id>
		<concept_desc>\Theory of computation~Process calculi</concept_desc>
		<concept_significance>300</concept_significance>
	</concept>
	<concept>
		<concept_id>10003752.10010124.10010131.10010134</concept_id>
		<concept_desc>Theory of computation~Operational semantics</concept_desc>
		<concept_significance>300</concept_significance>
	</concept>
	<concept>
		<concept_id>10003752.10010124.10010138</concept_id>
		<concept_desc>Theory of computation~Program reasoning</concept_desc>
		<concept_significance>300</concept_significance>
	</concept>
	<concept>
		<concept_id>10011007.10010940.10010971.10010980</concept_id>
		<concept_desc>Software and its engineering~Software system models</concept_desc>
		<concept_significance>300</concept_significance>
	</concept>
	<concept>
		<concept_id>10011007.10010940.10010992.10010993</concept_id>
		<concept_desc>Software and its engineering~Correctness</concept_desc>
		<concept_significance>300</concept_significance>
	</concept>
	<concept>
		<concept_id>10011007.10011006.10011039.10011311</concept_id>
		<concept_desc>Software and its engineering~Semantics</concept_desc>
		<concept_significance>300</concept_significance>
	</concept>
	<concept>
		<concept_id>10002944.10011122.10002945</concept_id>
		<concept_desc>General and reference~Surveys and overviews</concept_desc>
		<concept_significance>100</concept_significance>
	</concept>
</ccs2012>
\end{CCSXML}

% End generated code


\thanks{%
  Authors' addresses: 
  R.J.~van~Glabbeek, P.~H\"ofner,
  DATA61, CSIRO, Locked Bag 6016, UNSW, Sydney,  NSW 1466, Australia.
}

\maketitle
\newcommand{\ie}{i.e.\xspace}
\section{Introduction}
\emph{Fairness properties} reduce the sets of infinite potential runs of systems.
They are used in specification and verification.

As part of a system specification, a fairness property augments a
core behavioural specification, given for instance in process algebra
or pseudocode. The core behavioural specification typically generates
a transition system, and as such determines a set of finite and infinite potential
runs of the specified system.
A fairness property disqualifies some
of the infinite runs. These `unfair' runs are unintended by the overall specification.
It is then up to the implementation to ensure that only fair runs can
occur. This typically involves design decisions from which the
specification chooses to abstract.

As part of a verification, a fairness property models an assumption on
the specified or implemented system.  Fairness assumptions are normally used when verifying
\emph{liveness properties}, saying that something good will eventually
happen. Without making the fairness assumption, the liveness property
may not hold.  When verifying a liveness property of a specification
under a fairness assumption, this guarantees that the liveness property
holds for any implementation that correctly captures the fairness
requirement. When verifying a liveness property of an implemented
system under a fairness assumption, the outcome is a conditional
guarantee, namely that the liveness property holds as long as the
system behaves fairly;  when the system does not behave fairly, all
bets are off. Such a conditional guarantee tells us that 
there are no
design flaws in the system other than the---often necessary---reliance on fairness.

\emph{Progress properties} reduce the sets of finite potential runs of
systems. They play the same role as fairness assumptions do for
infinite runs. In the verification of liveness properties, progress
assumptions are essential. Although many interesting liveness
properties hold without making fairness assumptions of any kind, no
interesting liveness property holds without assuming progress.  On the
other hand, whereas a fairness assumption may be far-fetched, in the
sense that run-of-the-mill implementations tend not to be fair,
progress holds in almost any context, and run-of-the-mill
implementations almost always satisfy the required progress
properties. For this reason, progress assumptions are often made
implicitly, without acknowledgement that the proven liveness property
actually depends on the validity of such an assumption.

One contribution of this paper is a taxonomy of progress and fairness
properties found in the literature. These 
are ordered by their
strength in ruling out potential runs of systems, and thereby their
strength as an aid in proving liveness properties.
Our classification includes the  `classical' notions of fairness, 
such as weak and strong fairness, fairness of actions, transitions, instructions and components,
as well as extreme fairness. We also include \emph{probabilistic fairness}
(cf.~\cite{Pn83}) and what we call \emph{full fairness} \cite{BK86,BBK87a}.
These are methods to obtain liveness properties, that---like fairness
assumptions---do not require the liveness property to hold for each
infinite potential run. They 
differ from `classical' fairness assumptions in that
no specific set of potential runs is ruled out.

Another contribution is the introduction of the concept 
of \emph{justness}.  Justness assumptions form a middle ground  between
progress and fairness assumptions. They rule out a collection of
infinite runs, rather than finite ones, but based on a criterion that
is a progress assumption in a distributed context taken to its
logical conclusion.  Justness is weaker than most fairness assumptions, in the sense
that fewer runs are ruled out. It is more fundamental than fairness, in
the sense that many liveness properties of realistic systems do not
hold without assuming at least justness, whereas a fairness assumption
would be overkill.  As for progress, run-of-the-mill implementations
almost always satisfy the required justness properties.

The paper is organised as follows.
\Sec{liveness} briefly recapitulates transitions systems and  the concept of liveness (properties). 
\Sec{progress} recalls the assumption of progress, saying that a system will not get stuck without reason.
We recapitulate the classical notions of weak and strong fairness in \Sec{fairness}, and formulate them in the general setting of transition systems.
In \Sec{taxonomy} we capture as many as twelve notions of fairness found in the literature as instances
of a uniform definition, and order them in a hierarchy.
We evaluate these notions of fairness against three criteria from the literature, and one new one, in \Sec{criteria}. 
In \Sec{sw} we extend our taxonomy by strong weak fairness, an intermediate concept between weak and strong fairness. 
We argue that such a notion is needed as in realistic scenarios weak fairness can be too weak, and strong fairness too strong.
\Sec{ltl} recalls the formulations of weak and strong fairness in linear-time temporal logic. 
In \Sec{full} we present the strongest notion of fairness conceivable, and call it
full fairness. \Sec{SFTransitions} proves that on finite-state transition systems, strong
fairness of transitions, one of the notions classified in \Sec{taxonomy}, is as strong as full fairness.
\Sects~\ref{sec:probabilistic} and \ref{sec:extreme} relate our taxonomy to two further
  concepts: probabilistic and extreme fairness. We show that 
  probabilistic fairness is basically identical to strong fairness of transitions,
  whereas extreme fairness is the strongest possible notion of fairness that (unlike full fairness)
  fits the unifying definition proposed in \Sec{fairness}.
\Sec{extreme} concludes the comparison of fairness notions found in the literature. 
We then argue, in \Sec{justness}, that the careless use of any fairness assumption in verification tasks can yield false results. 
To compensate for this deficiency, we discuss a stronger version of progress, called justness, which can be used without qualms.
\Sec{J-fairness} develops fairness notions based on justness, and adds them to the hierarchy of fairness notions.
\Sec{events} discusses fairness of events, another notion of fairness from the literature, and shows that, although
  defined in a rather different way, it coincides with justness.
Sections~\ref{sec:liveness}--\ref{sec:events} deal with closed systems,
having no run-time interactions with the environment. This allows us a simplified treatment
  that is closer to the existing literature.
Section~\ref{sec:reactive} generalises almost all definitions and results to reactive systems,
interacting with their environments through synchronisation of actions.
This makes the fairness notions more widely applicable. 
The penultimate \Sec{eval} evaluates the notions of fairness discussed in the paper against the criteria for appraising fairness properties.
We relate our hierarchy to earlier classifications of fairness assumptions found in the literature in \Sec{conclusion}, and close with a short 
discussion about future work.

\section{Transition Systems and Liveness Properties}\label{sec:liveness}

In order to formally define liveness as well as progress, justness and fairness
properties,  we take transition systems as our system model.
Most other specification formalisms, such as process algebras, pseudocode, or Petri nets, have a
straightforward interpretation in terms of transition systems. Using this, any liveness,
progress, justness or fairness property defined on transition systems is also defined in terms of
those specification formalisms.

Depending on the type of property, the transition systems need various augmentations.
We introduce these as needed, starting with basic transition systems.

\renewcommand{\L}{\ell}
\begin{definition}\label{df:TS}
A \emph{transition system} is a tuple $G=(S, \Tr, \source,\target,I)$ with $S$ and $\Tr$ sets
(of \emph{states} and \emph{transitions}),
$\source,\target:\Tr\rightarrow S$,
and $I\subseteq S$ (the \emph{initial} states). 
\end{definition}
Later we will augment this definition with various attributes of transitions, such that
different transitions between the same two states may have different attributes.
It is for this reason that we do not simply 
introduce transitions as ordered pairs of states.
One of the simplest forms of such an augmentation is a \emph{labelled transition system}, which
features an additional function $\L$ mapping transitions into some set of labels.

\begin{exam}\label{ex:progress}
The program $x:=1;~y:=3$ is represented by the labelled transition system
$(\{1,2,3\}, \{t_1,t_2\}, \source, \target, \{1\})$, where
$\source(t_1)=1$, $\target(t_1)=2=\source(t_2)$, \mbox{$\target(t_2)=3$}; the labelling is given by 
$\L(t_1) = (x:=1)$ and $\L(t_2) = (y:=3)$. Any transition system (including any augmentation) can be
  represented diagrammatically. The program $x:=1;~y:=3$ is depicted by the following diagram.
\\
{\expandafter\ifx\csname graph\endcsname\relax
   \csname newbox\expandafter\endcsname\csname graph\endcsname
\fi
\ifx\graphtemp\undefined
  \csname newdimen\endcsname\graphtemp
\fi
\expandafter\setbox\csname graph\endcsname
 =\vtop{\vskip 0pt\hbox{%
\pdfliteral{
q [] 0 d 1 J 1 j
0.576 w
0.576 w
28.8 -4.32 m
28.8 -6.70587 26.86587 -8.64 24.48 -8.64 c
22.09413 -8.64 20.16 -6.70587 20.16 -4.32 c
20.16 -1.93413 22.09413 0 24.48 0 c
26.86587 0 28.8 -1.93413 28.8 -4.32 c
S
Q
}%
    \graphtemp=.5ex
    \advance\graphtemp by 0.060in
    \rlap{\kern 0.340in\lower\graphtemp\hbox to 0pt{\hss {\scriptsize 1}\hss}}%
\pdfliteral{
q [] 0 d 1 J 1 j
0.576 w
0.072 w
q 0 g
12.96 -2.52 m
20.16 -4.32 l
12.96 -6.12 l
12.96 -2.52 l
B Q
0.576 w
0 -4.32 m
12.96 -4.32 l
S
100.8 -4.32 m
100.8 -6.70587 98.86587 -8.64 96.48 -8.64 c
94.09413 -8.64 92.16 -6.70587 92.16 -4.32 c
92.16 -1.93413 94.09413 0 96.48 0 c
98.86587 0 100.8 -1.93413 100.8 -4.32 c
S
Q
}%
    \graphtemp=.5ex
    \advance\graphtemp by 0.060in
    \rlap{\kern 1.340in\lower\graphtemp\hbox to 0pt{\hss {\scriptsize 2}\hss}}%
\pdfliteral{
q [] 0 d 1 J 1 j
0.576 w
0.072 w
q 0 g
84.96 -2.52 m
92.16 -4.32 l
84.96 -6.12 l
84.96 -2.52 l
B Q
0.576 w
28.8 -4.32 m
84.96 -4.32 l
S
Q
}%
    \graphtemp=\baselineskip
    \multiply\graphtemp by -1
    \divide\graphtemp by 2
    \advance\graphtemp by .5ex
    \advance\graphtemp by 0.060in
    \rlap{\kern 0.840in\lower\graphtemp\hbox to 0pt{\hss $x:=1$\hss}}%
\pdfliteral{
q [] 0 d 1 J 1 j
0.576 w
172.8 -4.32 m
172.8 -6.70587 170.86587 -8.64 168.48 -8.64 c
166.09413 -8.64 164.16 -6.70587 164.16 -4.32 c
164.16 -1.93413 166.09413 0 168.48 0 c
170.86587 0 172.8 -1.93413 172.8 -4.32 c
h q 0.7 g
B Q
Q
}%
    \graphtemp=.5ex
    \advance\graphtemp by 0.060in
    \rlap{\kern 2.340in\lower\graphtemp\hbox to 0pt{\hss {\scriptsize 3}\hss}}%
\pdfliteral{
q [] 0 d 1 J 1 j
0.576 w
0.072 w
q 0 g
156.96 -2.52 m
164.16 -4.32 l
156.96 -6.12 l
156.96 -2.52 l
B Q
0.576 w
100.8 -4.32 m
156.96 -4.32 l
S
Q
}%
    \graphtemp=\baselineskip
    \multiply\graphtemp by -1
    \divide\graphtemp by 2
    \advance\graphtemp by .5ex
    \advance\graphtemp by 0.060in
    \rlap{\kern 1.840in\lower\graphtemp\hbox to 0pt{\hss $y:=3$\hss}}%
    \hbox{\vrule depth0.120in width0pt height 0pt}%
    \kern 2.400in
  }%
}%

\centerline{\raisebox{3ex}{\box\graph}~.}}\\[1ex]
\noindent
Here the short arrow without a source state indicates an initial state; we do not present the names of the transitions.
\end{exam}
Progress and fairness assumptions are often made when verifying liveness properties.
A \emph{liveness property} says that ``something [good] must happen'' eventually \cite{Lam77}.
One of the ways to formalise this in terms of transition systems is to postulate a set of good
states $\Ge \subseteq S$. In \ex{progress}, 
for instance, 
the good thing could be $y=3$, so that $\Ge$
consists of state 3 only---indicated by shading.  The liveness property given by $\Ge$ is
now 
defined to hold iff each system run---a concept yet to be formalised---reaches a state~$g\in\Ge$.

We now formalise a potential run of a system as a \emph{rooted path} in its representation as a
transition system.

\begin{definition}
A \emph{path} in a transition system $G=(S,\Tr,$ $\source,\target,I)$ is an alternating
sequence $\pi = s_0\,t_1\,s_1\,t_2\,s_2\cdots$
of states and transitions, starting with a state and either being infinite or ending with a state,
such that
$\source(t_i)=s_{i-1}$ and $\target(t_i)=s_i$ for all relevant $i$; it is \emph{rooted} if $s_0\in I$.
\end{definition}
Example \ref{ex:progress} has three potential runs, represented by the rooted paths $1\,t\,2\,u\,3$, $1\,t\,2$ and $1$,
where $t$ and $u$ denote the two transitions corresponding to the instructions $x:=1$ and $y:=3$.
The rooted path $1\,t\,2$ models a run consisting of the
execution of $x:=1$ only, without ever following this up with the instruction $y:=3$. In that run the
system stagnates in state 2. Likewise, the rooted path $1$ models a run where nothing ever happens.

Including these potential system runs as actual runs leads to a model of concurrency in which no
interesting liveness property ever holds, e.g. the liveness property $\Ge$ for \ex{progress}.%
\footnote{A \emph{partial run} is an initial part of a system run. Naturally, the paths
  $1\,t\,2$ and $1$ always model partial runs in \ex{progress}, even when we exclude them as runs.
  Partial runs play no role in this paper.}

Progress, justness and fairness assumptions are all used to
exclude some of the rooted paths as system runs. Each of them
can be formulated as, or gives rise to, a predicate on paths in transition systems.
We call a path \emph{progressing}, \emph{just} or \emph{fair}, if it meets the progress, justness or
fairness assumption currently under consideration. We call it \emph{complete} if it meets all
progress, justness and fairness assumptions we currently impose, so that
a rooted path is complete iff it models a run of the represented system that can actually occur.

\section{Progress}\label{sec:progress}

If the rooted path $1\,t\,2$ of \ex{progress} is not excluded and models an actual system 
run, the system is able to stagnate in state~2, and the program does not satisfy the
liveness property $\Ge$.
The assumption of \emph{progress}
excludes this type of behaviour.\footnote{{\sc Misra} \cite{Mis88,Mis01} calls this the `minimal
  progress assumption', attributed to {\sc Dijkstra}. The related `maximal progress assumption'
  \cite{Wa91}, often used in timed models, says that `if an agent can proceed by performing $\tau$-actions, then it should never wait unnecessarily'.
  The term `progress' appears in {\sc Gouda, Manning \& Yu} \cite{GMY84}, where a `progress state'
  is a state in which further transitions are enabled.

  In \cite{Mis01} {\sc Misra} uses `progress' as a synonym for `liveness'. 
  In session types, `progress' and 
  `global progress' are used as names of particular liveness properties \cite{CDPY13}; this use
  has no relation with ours.}
For closed systems, it says that
\begin{myquote}{
a (transition) system in a state that admits a transition will eventually progress, \ie perform a transition.}
\end{myquote}
In other words, a run will never get stuck in a state with outgoing transitions.

\begin{definition}
A path in a transition system representing a closed system is \emph{progressing} if
either it is infinite or its last state is the source of no transition.
\end{definition}
When assuming progress---one only considers progressing paths---the program of \ex{progress} satisfies $\Ge$.

Progress is assumed in almost all work on process algebra that deals with liveness properties, mostly implicitly. 
{\sc Milner} makes an explicit progress assumption for the process algebra CCS in \cite{Mi80}.
%As evidence we refer to Section 1.5 "Unobservable actions" on page 15
%and 16 of [29]: "But what does the [tau]-transition mean, in our more
%mechanistic interpretation? In means that R in state r1 ... *may* at
%any time move silently to r1', and that if a b-experiment is never
%attempted it *will* do so." (Here b is the other outgoing transition of r1.)
A progress assumption is built into the temporal logics LTL \cite{Pn77}, CTL \cite{EC82} and CTL* \cite{EH86},
namely by disallowing states without outgoing transitions and evaluating temporal formulas
by quantifying over infinite paths only.%
\footnote{Exceptionally, states without outgoing transitions are allowed, and then
    quantification is over all \emph{maximal} paths, \ie\ paths that are infinite or end in a state
    without outgoing transitions \cite{DV95}.}
In \cite{KdR83} the `multiprogramming axiom' is a
progress assumption, whereas in \cite{AFK88} progress is assumed as a `fundamental liveness property'.
In many other papers on fairness a progress assumption is made implicitly~\cite{AO83}.

An assumption related to progress is that `atomic actions always terminate' \cite{OL82}.  Here
we have built this assumption into our definition of a path. To model runs in which the last
transition never terminates, one could allow  paths to end in a transition rather than a state.

\section{Fairness}\label{sec:fairness}

\begin{exam}\label{ex:weak}
Consider the following program \cite{LPS81}.
\[
\begin{array}{@{}c@{}}
\mbox{}\hspace{-4.5cm}\init{y}{1}\\[0.8ex]
\left.
\begin{array}{@{}l@{}}
y:=0
\end{array}
~~~
\right \|
~~~
\left.
\begin{array}{@{}l@{}}
{\bf while~}(y>0){\bf ~do~~~}y:=y+1{\bf ~~~od}
\end{array}
\right.
\end{array}
\]

\noindent This program runs two parallel threads:
the first simply sets $y$ to $0$, while the second one repeatedly increments $y$, as long as $y>0$.
We assume that in the second thread, the evaluation of the guard $y>0$ and the assignment
 $y:=y+1$ happen in one atomic step, so that it is not possible that first the guard is evaluated (positively),
 then $y:=0$ is executed, and subsequently $y:=y+1$.
 
{\makeatletter
\let\par\@@par
\par\parshape0
\everypar{}\begin{wrapfigure}[4]{r}{0.45\textwidth}
 \vspace{-.5ex}%
  \input{weak}%
  \centerline{\hspace{5mm}\raisebox{1ex}{\box\graph}}%
   \vspace{-15ex}
  \end{wrapfigure}
The transition system on the right describes the behaviour of this program.
Since \init{y}{1} is an
initialisation step, it is not shown in the transition system.
In the first state $y$ has a positive value, whereas in the second it is $0$.
\par}

A different transition system describing the same program is shown below.\\
{\footnotesize \expandafter\ifx\csname graph\endcsname\relax
   \csname newbox\expandafter\endcsname\csname graph\endcsname
\fi
\ifx\graphtemp\undefined
  \csname newdimen\endcsname\graphtemp
\fi
\expandafter\setbox\csname graph\endcsname
 =\vtop{\vskip 0pt\hbox{%
\pdfliteral{
q [] 0 d 1 J 1 j
0.576 w
0.576 w
26.208 -3.528 m
26.208 -5.476461 24.628461 -7.056 22.68 -7.056 c
20.731539 -7.056 19.152 -5.476461 19.152 -3.528 c
19.152 -1.579539 20.731539 0 22.68 0 c
24.628461 0 26.208 -1.579539 26.208 -3.528 c
S
Q
}%
    \graphtemp=.5ex
    \advance\graphtemp by 0.049in
    \rlap{\kern 0.000in\lower\graphtemp\hbox to 0pt{\hss ~\hss}}%
\pdfliteral{
q [] 0 d 1 J 1 j
0.576 w
0.072 w
q 0 g
11.952 -1.728 m
19.152 -3.528 l
11.952 -5.328 l
11.952 -1.728 l
B Q
0.576 w
0.72 -3.528 m
11.952 -3.528 l
S
82.872 -3.528 m
82.872 -5.476461 81.292461 -7.056 79.344 -7.056 c
77.395539 -7.056 75.816 -5.476461 75.816 -3.528 c
75.816 -1.579539 77.395539 0 79.344 0 c
81.292461 0 82.872 -1.579539 82.872 -3.528 c
S
Q
}%
\color{red}%
\pdfliteral{
q [] 0 d 1 J 1 j
0.576 w
26.208 -3.528 m
75.816 -3.528 l
S
Q
}%
    \graphtemp=\baselineskip
    \multiply\graphtemp by -1
    \divide\graphtemp by 2
    \advance\graphtemp by .5ex
    \advance\graphtemp by 0.049in
    \rlap{\kern 0.709in\lower\graphtemp\hbox to 0pt{\hss $y:=y{+}1$\hss}}%
\color{black}%
\pdfliteral{
q [] 0 d 1 J 1 j
0.576 w
0.072 w
q 0 g
68.616 -1.728 m
75.816 -3.528 l
68.616 -5.328 l
68.616 -1.728 l
B Q
0.576 w
74.376 -3.528 m
68.616 -3.528 l
S
26.208 -26.496 m
26.208 -28.444461 24.628461 -30.024 22.68 -30.024 c
20.731539 -30.024 19.152 -28.444461 19.152 -26.496 c
19.152 -24.547539 20.731539 -22.968 22.68 -22.968 c
24.628461 -22.968 26.208 -24.547539 26.208 -26.496 c
h q 0.5 g
B Q
0.072 w
q 0 g
24.48 -15.768 m
22.68 -22.968 l
20.88 -15.768 l
24.48 -15.768 l
B Q
0.576 w
22.68 -7.056 m
22.68 -15.768 l
S
Q
}%
    \graphtemp=.5ex
    \advance\graphtemp by 0.209in
    \rlap{\kern 0.315in\lower\graphtemp\hbox to 0pt{\hss $\hspace{28pt}y:=0$\hss}}%
\pdfliteral{
q [] 0 d 1 J 1 j
0.576 w
139.608 -3.528 m
139.608 -5.476461 138.028461 -7.056 136.08 -7.056 c
134.131539 -7.056 132.552 -5.476461 132.552 -3.528 c
132.552 -1.579539 134.131539 0 136.08 0 c
138.028461 0 139.608 -1.579539 139.608 -3.528 c
S
Q
}%
\color{red}%
\pdfliteral{
q [] 0 d 1 J 1 j
0.576 w
82.944 -3.528 m
132.552 -3.528 l
S
Q
}%
    \graphtemp=\baselineskip
    \multiply\graphtemp by -1
    \divide\graphtemp by 2
    \advance\graphtemp by .5ex
    \advance\graphtemp by 0.049in
    \rlap{\kern 1.496in\lower\graphtemp\hbox to 0pt{\hss $y:=y{+}1$\hss}}%
\color{black}%
\pdfliteral{
q [] 0 d 1 J 1 j
0.576 w
0.072 w
q 0 g
125.352 -1.728 m
132.552 -3.528 l
125.352 -5.328 l
125.352 -1.728 l
B Q
0.576 w
131.112 -3.528 m
125.352 -3.528 l
S
82.872 -26.496 m
82.872 -28.444461 81.292461 -30.024 79.344 -30.024 c
77.395539 -30.024 75.816 -28.444461 75.816 -26.496 c
75.816 -24.547539 77.395539 -22.968 79.344 -22.968 c
81.292461 -22.968 82.872 -24.547539 82.872 -26.496 c
h q 0.5 g
B Q
0.072 w
q 0 g
81.144 -15.768 m
79.344 -22.968 l
77.544 -15.768 l
81.144 -15.768 l
B Q
0.576 w
79.344 -7.056 m
79.344 -15.768 l
S
Q
}%
    \graphtemp=.5ex
    \advance\graphtemp by 0.209in
    \rlap{\kern 1.102in\lower\graphtemp\hbox to 0pt{\hss $\hspace{28pt}y:=0$\hss}}%
\pdfliteral{
q [] 0 d 1 J 1 j
0.576 w
196.272 -3.528 m
196.272 -5.476461 194.692461 -7.056 192.744 -7.056 c
190.795539 -7.056 189.216 -5.476461 189.216 -3.528 c
189.216 -1.579539 190.795539 0 192.744 0 c
194.692461 0 196.272 -1.579539 196.272 -3.528 c
S
Q
}%
\color{red}%
\pdfliteral{
q [] 0 d 1 J 1 j
0.576 w
139.608 -3.528 m
189.216 -3.528 l
S
Q
}%
    \graphtemp=\baselineskip
    \multiply\graphtemp by -1
    \divide\graphtemp by 2
    \advance\graphtemp by .5ex
    \advance\graphtemp by 0.049in
    \rlap{\kern 2.283in\lower\graphtemp\hbox to 0pt{\hss $y:=y{+}1$\hss}}%
\color{black}%
\pdfliteral{
q [] 0 d 1 J 1 j
0.576 w
0.072 w
q 0 g
182.016 -1.728 m
189.216 -3.528 l
182.016 -5.328 l
182.016 -1.728 l
B Q
0.576 w
187.776 -3.528 m
182.016 -3.528 l
S
139.608 -26.496 m
139.608 -28.444461 138.028461 -30.024 136.08 -30.024 c
134.131539 -30.024 132.552 -28.444461 132.552 -26.496 c
132.552 -24.547539 134.131539 -22.968 136.08 -22.968 c
138.028461 -22.968 139.608 -24.547539 139.608 -26.496 c
h q 0.5 g
B Q
0.072 w
q 0 g
137.88 -15.768 m
136.08 -22.968 l
134.28 -15.768 l
137.88 -15.768 l
B Q
0.576 w
136.08 -7.056 m
136.08 -15.768 l
S
Q
}%
    \graphtemp=.5ex
    \advance\graphtemp by 0.209in
    \rlap{\kern 1.890in\lower\graphtemp\hbox to 0pt{\hss $\hspace{28pt}y:=0$\hss}}%
\pdfliteral{
q [] 0 d 1 J 1 j
0.576 w
253.008 -3.528 m
253.008 -5.476461 251.428461 -7.056 249.48 -7.056 c
247.531539 -7.056 245.952 -5.476461 245.952 -3.528 c
245.952 -1.579539 247.531539 0 249.48 0 c
251.428461 0 253.008 -1.579539 253.008 -3.528 c
S
Q
}%
\color{red}%
\pdfliteral{
q [] 0 d 1 J 1 j
0.576 w
196.272 -3.528 m
245.88 -3.528 l
S
Q
}%
    \graphtemp=\baselineskip
    \multiply\graphtemp by -1
    \divide\graphtemp by 2
    \advance\graphtemp by .5ex
    \advance\graphtemp by 0.049in
    \rlap{\kern 3.071in\lower\graphtemp\hbox to 0pt{\hss $y:=y{+}1$\hss}}%
\color{black}%
\pdfliteral{
q [] 0 d 1 J 1 j
0.576 w
0.072 w
q 0 g
238.68 -1.728 m
245.88 -3.528 l
238.68 -5.328 l
238.68 -1.728 l
B Q
0.576 w
244.512 -3.528 m
238.68 -3.528 l
S
196.272 -26.496 m
196.272 -28.444461 194.692461 -30.024 192.744 -30.024 c
190.795539 -30.024 189.216 -28.444461 189.216 -26.496 c
189.216 -24.547539 190.795539 -22.968 192.744 -22.968 c
194.692461 -22.968 196.272 -24.547539 196.272 -26.496 c
h q 0.5 g
B Q
0.072 w
q 0 g
194.544 -15.768 m
192.744 -22.968 l
190.944 -15.768 l
194.544 -15.768 l
B Q
0.576 w
192.744 -7.056 m
192.744 -15.768 l
S
Q
}%
    \graphtemp=.5ex
    \advance\graphtemp by 0.209in
    \rlap{\kern 2.677in\lower\graphtemp\hbox to 0pt{\hss $\hspace{28pt}y:=0$\hss}}%
\pdfliteral{
q [] 0 d 1 J 1 j
0.576 w
309.672 -3.528 m
309.672 -5.476461 308.092461 -7.056 306.144 -7.056 c
304.195539 -7.056 302.616 -5.476461 302.616 -3.528 c
302.616 -1.579539 304.195539 0 306.144 0 c
308.092461 0 309.672 -1.579539 309.672 -3.528 c
S
Q
}%
\color{red}%
\pdfliteral{
q [] 0 d 1 J 1 j
0.576 w
253.008 -3.528 m
302.616 -3.528 l
S
Q
}%
    \graphtemp=\baselineskip
    \multiply\graphtemp by -1
    \divide\graphtemp by 2
    \advance\graphtemp by .5ex
    \advance\graphtemp by 0.049in
    \rlap{\kern 3.858in\lower\graphtemp\hbox to 0pt{\hss $y:=y{+}1$\hss}}%
\color{black}%
\pdfliteral{
q [] 0 d 1 J 1 j
0.576 w
0.072 w
q 0 g
295.416 -1.728 m
302.616 -3.528 l
295.416 -5.328 l
295.416 -1.728 l
B Q
0.576 w
301.176 -3.528 m
295.416 -3.528 l
S
253.008 -26.496 m
253.008 -28.444461 251.428461 -30.024 249.48 -30.024 c
247.531539 -30.024 245.952 -28.444461 245.952 -26.496 c
245.952 -24.547539 247.531539 -22.968 249.48 -22.968 c
251.428461 -22.968 253.008 -24.547539 253.008 -26.496 c
h q 0.5 g
B Q
0.072 w
q 0 g
251.28 -15.768 m
249.48 -22.968 l
247.68 -15.768 l
251.28 -15.768 l
B Q
0.576 w
249.48 -7.056 m
249.48 -15.768 l
S
Q
}%
    \graphtemp=.5ex
    \advance\graphtemp by 0.209in
    \rlap{\kern 3.465in\lower\graphtemp\hbox to 0pt{\hss $\hspace{28pt}y:=0$\hss}}%
    \graphtemp=.5ex
    \advance\graphtemp by 0.049in
    \rlap{\kern 4.488in\lower\graphtemp\hbox to 0pt{\hss $\!\!\cdots$\hss}}%
    \hbox{\vrule depth0.417in width0pt height 0pt}%
    \kern 4.488in
  }%
}%

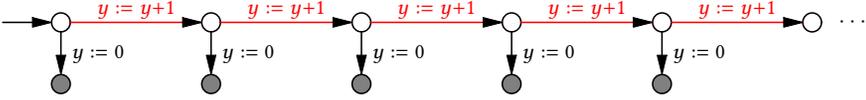
\centerline{\box\graph}}
\vspace{1.5ex}

\noindent
A liveness property is that eventually $y=0$.
It is formalised by letting $\Ge$ consist of the rightmost state in the first,
or of all the bottom states in the second transition system.
When assuming progress only, this liveness property does not hold; the only counterexample is
the red coloured path.
This path might be considered `unfair', because the left-hand thread is never executed.
A suitable fairness assumption rules out this path as a possible system run.
\end{exam}
To formalise fairness we use transition systems $G=(S,\Tr,\source,\target,I,\T)$
that are augmented with a set $\T\subseteq\Pow({\Tr})$ of \emph{tasks}
$T\subseteq {\Tr}$, each being a set of transitions.%

\begin{definition}\label{df:fair}
For a given transition system $G\mathbin=(S,\Tr,\source,\target,I,\T)$, a task $T\mathbin\in\T$
is~\emph{enabled} in a state $s\in S$ if there exists a transition $t\in T$ with $\source(t)=s$.
The task is said to be \emph{perpetually enabled} on a path $\pi$ in $G$, if it is enabled in every state of
$\pi$. It is \emph{relentlessly enabled} on $\pi$, if each suffix of $\pi$ contains a state
in which it is enabled.\footnote{This is the case if the task is enabled in infinitely many states
of $\pi$, in a state that occurs infinitely often in $\pi$, or in the last state of a finite $\pi$.}
It \emph{occurs} in $\pi$ if $\pi$ contains a transition $t\in T$.

A path $\pi$ in $G$ is \emph{weakly fair} if, for every suffix $\pi'$ of $\pi$,
each task that is perpetually enabled on $\pi'$, occurs in $\pi'$.

A path $\pi$ in $G$ is \emph{strongly fair} if, for every suffix $\pi'$ of $\pi$,
each task that is relentlessly  enabled on $\pi'$,
occurs~in~$\pi'$.
\end{definition}

The purpose of these properties is to rule out certain paths from consideration (the unfair ones),
namely those with a suffix in which a task is enabled perpetually (or, in the strong case, relentlessly)
yet never occurs.
To avoid the quantification over suffixes, these properties can be reformulated:

A path $\pi$ in $G$ is \emph{weakly fair} if each task that from some state onwards is perpetually enabled
on $\pi$, occurs infinitely often in $\pi$. 

A path $\pi$ in $G$ is \emph{strongly fair} if each task that is relentlessly  enabled on $\pi$,
occurs infinitely often in $\pi$.

Clearly, any path that is strongly fair, is also weakly fair.

\emph{Weak} [\emph{strong}] \emph{fairness} is the assumption that only weakly [strongly] fair
rooted paths represent system runs.
When applied to pseudocode, process algebra expressions or other system specifications that translate
to transition systems, the concepts of weak and strong fairness are parametrised by the way to
extract the collection $\T$ of tasks from the specification. For one such extraction,
called \emph{fairness of directions} in \Sec{taxonomy}, 
weak and strong fairness were introduced in \cite{AO83}. Weak fairness was introduced independently,
under the name \emph{justice}, in \cite{LPS81}, and strong fairness, under the name \emph{fairness},
in \cite{GFMdR85} and \cite{LPS81}.
Strong fairness is called \emph{compassion} in \cite{MP89,MP95}.

Recall that the transition systems of \ex{weak}, when merely assuming progress, do not satisfy
the liveness property $\Ge$.
To change that verdict, we can assume fairness, by defining a collection $\T$ of tasks.
We can, for instance, declare two tasks: $T_1$ being the set of all transitions labelled $y:=y+1$ and
$T_2$ the set of all transitions labelled $y:=0$. Now the red path in the second transition system, as well as all its prefixes,
becomes unfair, since task $T_2$ is perpetually enabled, yet never occurs.
Thus, in \ex{weak}, $\Ge$ does hold when assuming weak fairness, with $\T=\{T_1,T_2\}$.
In fact, it suffices to take $\T=\{T_2\}$.
Assuming strong fairness gives the same result.

\begin{exam}\label{ex:mutex}
\hfuzz 3pt
Consider the following basic mutual exclusion protocol \cite{LPS81}.
\vspace{.3ex}
\[
\begin{array}{@{}c@{}}
\init{y}{1}\\[1ex]
\left.
\begin{array}{@{}l@{}}
{\bf while}~(true)\\
\left\{\begin{array}{ll}
\ell_0 & {\bf noncritical~section}\\
\ell_1 & {\bf await} (y>0) \{ y:=y{-}1\}\\ 
\ell_2 & {\bf critical~section}\\
\ell_3 & y:=y{+}1
\end{array}\right.
\end{array}
~~~
\right \|
~~~
\begin{array}{@{}l@{}}
{\bf while}~(true)\\
\left\{\begin{array}{ll}
m_0 & {\bf noncritical~section}\\
m_1 & {\bf await} (y>0) \{ y:=y{-}1\}\\
m_2 & {\bf critical~section}\\
m_3 & y:=y{+}1
\end{array}\right.
\end{array}
\end{array}
\vspace{1ex}
\]

\noindent Here instructions $\ell_1$ and $m_1$ wait until $y>0$ (possibly
forever) and then atomically execute the instruction $y:=y-1$,
meaning without allowing the other process to change the value of $y$
between the evaluation $y>0$ and the assignment $y:=y-1$.
{\sc Dijkstra} \cite{Dijkstra65} abbreviates this instruction as P$(y)$; it
models entering of a critical section, with $y$ as semaphore.
Instructions $\ell_3$ and $m_3$ model leaving the critical section,
and are  abbreviated V$(y)$.
{\makeatletter
\let\par\@@par
\par\parshape0
\everypar{}\begin{wrapfigure}[6]{r}{0.35\textwidth}
 \vspace{-1ex}
 \input{mutex}
  \centerline{\raisebox{1ex}{\box\graph}}
 \end{wrapfigure}
The induced transition system is depicted on the right, when
ignoring lines $\ell_0$ and $m_0$, which play no significant role.

Let $\Ge$ contain the single marked state; so the good thing we hope to accomplish is the occurrence
of $\ell_{2}$, saying that the left-hand process executed its critical section.
Our tasks could be $L=\{\ell_1,\ell_2,\ell_3\}$ and $M=\{m_1,m_2,m_3\}$.
Here weak fairness is insufficient to ensure $\Ge$,
but under the assumption of strong fairness this liveness property holds.
\par}
\end{exam}

\noindent In \cite{GFMdR85,LPS81}, the assumptions of weak and strong fairness restrict the set of infinite runs only.
These papers consider a path $\pi$ in $G$ weakly fair if either it is finite, or each task that from
some state onwards is perpetually enabled on $\pi$, occurs infinitely often in $\pi$. 
Likewise, a path is considered strongly fair if either it is finite, or each task that is
relentlessly  enabled on $\pi$ occurs infinitely often in~$\pi$.
This makes it necessary to assume progress in addition to fairness when evaluating the validity of
liveness properties. 

In this paper we have dropped the exceptions for finite paths.
The effect of this is that the progress assumption is included in weak or strong fairness, at least
when $\bigcup_{T\in\T}T=\Tr$. The latter condition can always be realised by adding the task $\Tr$ to
$\T$---this has no effect on the resulting notions of weak and strong fairness for infinite paths.
Moreover, dropping the exceptions for finite paths poses no further restrictions on finite paths
beyond progress.

\section{A Taxonomy of Fairness Properties}\label{sec:taxonomy}

For a given transition system, the concepts of strong and weak fairness
are completely determined by the collection $\T$ of tasks.
We now address the question how to obtain $\T$.
The different notions of fairness found in the literature can be cast
as different answers to that question.

As a first classification we distinguish local and
global fairness properties. 

A \emph{global fairness} property is
formulated for transition systems that are not naively equipped with
a set $\T$, but may have some other auxiliary structure, inherited
from the system specification (such as program code or a process
algebra expression) that gave rise to that transition system. The set
$\T$ of tasks, and thereby the concepts of strong and weak fairness,
is then extracted in a systematic way from this auxiliary structure.
An example of this is \emph{fairness of actions}. Here transitions\pagebreak[3]
are labelled with \emph{actions}, activities that happen when the transition is taken.
In \ex{weak}, the actions are $\{y\mathbin{:=}y{+}1,~ y\mathbin{:=}0\}$, and in \ex{mutex}
$\{\ell_0,\ell_1,\ell_2,\ell_3,m_0,m_1,m_2,m_3\}$.
Each action constitutes a task, consisting of all the transitions
labelled with that action.
Formally, $\T = \{T_{a}\mid a\in Act\}$, where $Act$ is the set of actions and
$T_{a} = \{t\mid \mbox{$t$ is labelled by $a$}\}$.
This is in fact the notion of fairness
employed in \ex{weak}, and using it for \ex{mutex} would give
rise to the same set of fair paths as the set of tasks employed there.

{\makeatletter
\let\par\@@par
\par\parshape0
\everypar{}\begin{wrapfigure}[3]{r}{0.35\textwidth}
 \vspace{-2.9ex}
  \input{local}
  \centerline{\raisebox{1ex}{\box\graph}}
   \vspace{-10ex}
  \end{wrapfigure}
A \emph{local fairness} property on the other hand is created in an ad hoc
manner to add new information to a system specification. In the transition system
on the right,
for instance, we may want to make sure that incurring the cost $c$ is always
followed by reaping the benefit~$b$. 
This can be achieved by explicitly adding
a fairness property to the specification, ruling out any path that
gets stuck in the $d$-loop.
This can be accomplished by declaring $\{b\}$ a task.
At the same time, it may be acceptable  that we never escape from the $a$-loop.
Hence there is no need to declare a task $\{c\}$. Whether we also
declare tasks $\{a\}$ and $\{d\}$ makes no difference, provided that progress is
assumed---as we do.
A global fairness property seems inappropriate for this example, as it is
hard to imagine why it would treat the $a$-loop with the $c$-exit
differently from the $d$-loop with the $b$-exit.
\par}

Under local fairness, a system specification consists of two parts: a
core part that generates a transition system, and a fairness component.
This is the way fairness is incorporated in the specification language TLA$^+$ \cite{La02}.
It is also the form of fairness we have chosen to apply in \citep[Sect.~9]{TR13}
for the formal specification of the Ad hoc On-demand Distance Vector (AODV)
protocol~\cite{rfc3561}, a wireless mesh network routing protocol,
using the process algebra AWN \cite{FGHMPT12a} for the first part and Linear-time
Temporal Logic (LTL)~\cite{Pn77} for the second.
Under global fairness, a system specification consists of the core part
only, generating a transition system with some auxiliary structure,
together with just the choice of a notion of fairness, that extracts
the set of tasks---or directly the set of fair paths---from this
augmented transition system. Below we discuss some of these
global notions of fairness, and cast them as a function from the
appropriate auxiliary structure of transition systems to a set of tasks.

For the forthcoming examples we use the process algebra CCS~\cite{Mi80}.
Given a set of \emph{action names} $a$, $b$, $c$, \dots,
  where each action $a$ has a complement $\bar a$, 
the set of CCS processes is built from the following constructs:
${\bf 0}$ is the empty process (deadlock) offering no continuation;
the process $a.E$ can perform an action $a$ and continue as $E$;%
\footnote{Assuming progress, it \emph{has}  to perform the action.}\,%
\footnote{We often abbreviate $a.{\bf 0}$ by $a$.}
the process $\tau.E$ can perform the internal action $\tau$ and continue as $E$;
the process $E+F$ offers a choice and can proceed either as $E$ or as $F$; 
the parallel composition $E|F$ allows $E$ and $F$ to progress independently; moreover, 
in case $E$ can perform an action $a$ 
and $F$ can perform its complement $\bar a$ (or vice versa),
the two processes can synchronise
by together performing an internal step $\tau$; 
the process $E\backslash a$ modifies $E$ by inhibiting the
actions $a$ and $\bar a$; and
the relabelling $E[f]$ uses the function $f$ to rename action names in $E$.\vspace{1pt}
Infinite processes can be specified by process variables $X$ that are bound by
\emph{defining equations} \plat{$X\defis E$}, where $E$ is a process
expression constructed from the elements described above.
CCS expressions generate transition systems where transitions are labelled with the actions a process performs---in case of synchronisation the label is $\tau$.
The complete formal syntax and semantics is presented in \App{CCS}.

{\makeatletter
\let\par\@@par
\par\parshape0
\everypar{}\begin{wrapfigure}[4]{r}{0.3\textwidth}
 \vspace{-2ex}
   \input{uniform}
   \centerline{\raisebox{1ex}{\box\graph}}
   \vspace{-10ex}
  \end{wrapfigure}
We illustrate the following definitions of global fairness by a uniform example.
Consider the process $X | Y$ where  \plat{$X \defis a.X+b.X$} and \plat{$Y \defis a.Y+\bar b.Y$}.\,\, 
Its transition system is depicted on the right, where the names $t_i$ differentiate transitions.
The specification contains two occurrences of $a$, which we will call $a_1$ and $a_2$.
\pagebreak[3]

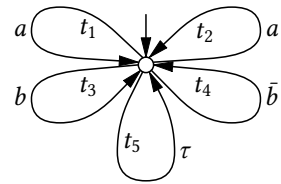
\begin{description}
\item{{\bf Fairness of actions} (A)} 
\emph{Fairness of actions} is already defined above; the auxiliary structure it requires is a
labelling $\ell:\Tr{\,\rightarrow} Act$ of transitions by actions.
For our running example we have $\ell(t_1)=\ell(t_2)=a$, $\ell(t_3)=b$, $\ell(t_4)=\bar b$ and $\ell(t_5)=\tau$.
This notion of fairness appears, for instance, in \cite{La00} and \citep[Def.~3.44]{BK08}.
It has been introduced, under the name {fairness of transitions}, in \citep[Def.~1]{QS83}.

\item{{\bf Fairness of transitions} (T)}
Taking the tasks to be the transitions, rather than their labels,
 is another way to define a collection $\T$ of tasks for a given transition system $G$. 
 $\T$ consists of all singleton sets of transitions of $G$.
In the first transition system of \ex{weak} $\T$ would consist of two elements, 
whereas the second transition system would yield an infinite set $\T$.
This notion of fairness does not require any
auxiliary structure on $G$. Fairness of transitions appears in \citep[Page 127]{Fr86} under the name
{$\sigma$-fairness}. A small variation, obtained by identifying all transitions that have the same
source and target, originates from \citep[Def.~2]{QS83}.
\end{description}
We now look at a transition system where each transition is labelled
with the nonempty set of \emph{instructions} in the program text that gave rise
to that transition.
So the auxiliary structure is a function $\instr:\Tr\rightarrow\Pow(\I)$ with
$\instr(t)\mathbin{\neq}\emptyset$ for all $t\mathbin\in\Tr$, where 
 $\I$ is a set (of \emph{instructions}).
The set $\instr(t)$ has two elements for a handshake communication and
N for N-way communication.
For our example we have $\I=\{a_1,a_2,b,\bar b\}$, with $\instr(t_1)=\{a_1\}$, $\instr(t_2)=\{a_2\}$, $\instr(t_3)=\{b\}$, $\instr(t_4)=\{\bar b\}$ and 
$\instr(t_5)=\{b,\bar b\}$.
\begin{description}

\item{{\bf Fairness of instructions} (I)}
This notion assigns a task to each single instruction;
a transition belongs to that task if that instruction contributed to it.
So $\T\mathrel{:=} \{T_I \mid I\mathop\in\I\}$ with $T_I\mathbin{:=}\{t\mathop\in\Tr \mid I\mathop\in\instr(t)\}$.
In the example, $T_{a_1}=\{t_1\}$, $T_{a_2}=\{t_2\}$, $T_{b}=\{t_3,t_5\}$ and $T_{\bar b}=\{t_4,t_5\}$.
This type of fairness appears in \cite{KdR83} under the name {guard fairness}.

\item{{\bf Fairness of synchronisations} (\Sy)}
\phantomsection{}
\label{sy_synchronisation}
We can also use $\instr{}$ to assign a task to each set of instructions; a transition belongs to that task if it is
obtained through synchronisation of exactly that set of instructions:
$\T:= \{T_Z \mid Z\subseteq\I\}$ with $T_Z:=\{t\in\Tr \mid \instr(t)=Z\}$.
Here a local action, involving one component only, counts as a singleton synchronisation.
In the example, $T_{\{a_1\}}=\{t_1\}$, $T_{\{a_2\}}=\{t_2\}$, $T_{\{b\}}=\{t_3\}$, $T_{\{\bar b\}}=\{t_4\}$ and $T_{\{b, \bar b\}}=\{t_5\}$;
all other tasks, such as $T_{\{a_1,a_2,b\}}$, are empty.
This type of fairness appears in \cite{KdR83} under the name {channel fairness},
and in \cite{AFK88} under the name {communication fairness}.

\end{description}
For the remaining two notions of fairness
we need a transition system where each transition is labelled 
with the parallel \emph{component} of the represented system that gave rise
to that transition.\footnote{\label{minimal}%
A parallel component $c$ may itself have further parallel
components $c_1$ and $c_2$, so that the collection of all parallel components in  a
given state of the system may be represented as a tree (with the entire system in that state as root). A transition that can be
attributed to component $c_1$ thereby implicitly also belongs to $c$.
We label such a transition with $c_1$ only, thus employing merely the leaves in such a tree.
The number of components of a system is dynamic, as a single component may split into
multiple components only after executing some instructions.}
When allowing synchronisation, each transition stems from a nonempty set of
components. So the auxiliary structure is a function $\comp:\Tr\rightarrow\Pow(\Ce)$ with
$\comp(t)\mathbin{\neq}\emptyset$ for all $t\mathbin\in\Tr$.
Our example has two components only, named $\Left$ (left) and $\R$ (right), respectively.
We have $\comp(t_1)=\comp(t_3)=\{\Left\}$, $\comp(t_2)=\comp(t_4)=\{\R\}$ and $\comp(t_5)=\{\Left,\R\}$.
\begin{description}
\item{{\bf Fairness of components} (C)}
Each component determines a task. A transition belongs to that
task if that component contributed to it.
So $\T:= \{T_C \mid C\in\Ce\}$ with $T_C:=\{t\in\Tr \mid C\in\comp(t)\}$.
In our example we have $T_{\Left}=\{t_1,t_3,t_5\}$ and $T_{\R}=\{t_2,t_4,t_5\}$.
This is the type of fairness studied in \cite{CS87,CDV06c}; it also
appears in \cite{KdR83,AFK88} under the name {process fairness}.

\item{{\bf Fairness of groups of components} (G)}
This is like fairness of components, except that each \emph{set} of
components forms a task, and a transition belongs to that task if it is
obtained through synchronisation of exactly that set of components:
$\T:= \{T_G \mid G\subseteq\Ce\}$ with $T_G:=\{t\mathbin\in\Tr \mid \comp(t)\mathbin=G\}$.
In our example we have $T_{\{\Left\}}=\{t_1,t_3\}$, $T_{\{\R\}}=\{t_2,t_4\}$ and $T_{\{\Left,\R\}}=\{t_5\}$. 
This type of fairness appears in \cite{AFK88} under the name {channel fairness} when
allowing only handshake communication, and under the name {group fairness} 
when allowing N-way communication.
\end{description}

Under each of the above notions of fairness, the first transition system of \ex{weak} has two tasks,
each consisting of one transition. The transition system of \ex{mutex} has six tasks of one
transition each for fairness of actions, transitions, instructions or synchronisations,
but two tasks under fairness of components or groups of components. For \ex{mutex}, $\Ge$ holds in all cases when
assuming strong fairness, but not when assuming weak fairness.

For transition systems derived from specification formalisms that do not feature synchronisation \cite{AO83,GFMdR85,LPS81},
fairness of instructions (I) and of synchronisations (Z) coincide.
This combined notion of fairness, which we call 
{\bf fairness of directions} (D), is the central notion of fairness in {\sc Francez} \cite{Fr86};
it is also the notion of fairness introduced in \cite{AO83,GFMdR85}.
As for the notions (I) and (Z),
we need a transition system where each transition is labelled
with the \emph{instructions} in the program text that gave rise
to that transition; however, here we assume that each transitions stems from only one
instruction (called \emph{direction} in \cite{Fr86}), so that the function $\instr$ can be
understood to have type $\Tr\rightarrow\I$. Each instruction gives rise to a task: $\T:= \{T_I \mid I\in\I\}$ with
$T_I:=\{t\in\Tr \mid \instr(t)=I\}$.
If we would take $Act := \I$ and $\ell := \instr$ then fairness of actions becomes the same as
fairness of directions. Since $\D$-fairness coincides with both $\In$ and $\Sy$-fairness whenever it is defined, 
\ie if $\instr:\Tr\rightarrow\Pow(\I)$ maps to singleton sets only,
we do not examine this notion separately. 

As for (I) and (Z), fairness of components (C) coincides with fairness of groups of components (G) for
transition systems derived from specification formalisms that do not feature synchronisation}.
This combination of (C) and (G) is the notion of fairness introduced in \cite{LPS81}.

In \App{fragment} we define the fragment of CCS used in this paper, and augment
the transition systems generated by this fragment
with the functions $\instr:\Tr\rightarrow\Pow(\I)$ and $\comp:\Tr\rightarrow\Pow(\Ce)$. 
Based on these formal definitions, the remainder of this section presents examples that distinguish each fairness notion from all the others.
In these examples we often consider a liveness property $\Ge$, 
indicated by the shaded states in the depicted transition systems.

Let $\X\Y$ be the fairness assumption with $\X \in \{\W,\St\}$ indicating weak or strong and
$\Y\in\{\Ac,\Ts,\In,\Sy,\linebreak[1]\Cp,\Gr\}$ one of the notions of fairness above.

\vspace{-1mm}
\begin{exam}\label{ex:SA}
Consider the process $a | X$ where $X \defis a.X$. Its

{\makeatletter
\let\par\@@par
\par\parshape0
\everypar{}\begin{wrapfigure}[3]{r}{0.35\textwidth}
 \vspace{-5ex}
  \input{SA}
  \centerline{\raisebox{1ex}{\box\graph}}
   \vspace{-10ex}
  \end{wrapfigure}
\noindent  transition system is depicted on  the
 right, where the names $t_i$ differentiate transitions.
The process specification contains two occurrences of the action $a$, which we will call $a_1$ and $a_2$.
These will be our instructions: $\instr:\Tr\rightarrow\Pow(\I)$ is given by
$\instr(t_1)=\instr(t_3)\mathbin=\{a_2\}$ and $\instr(t_2)\mathbin=\{a_1\}$.
The process has three components, namely $a$ (named $\Left$), $X$ (named $\R$), and the entire expression.
Now $\comp:\Tr\rightarrow\Pow(\Ce)$ is given by
$\comp(t_1)=\comp(t_3)=\{\R\}$ and $\comp(t_2)=\{\Left\}$.
Each of the notions of fairness I, \Sy, C and G yields two tasks: $\{t_1,t_3\}$ and $\{t_2\}$.
On the only infinite path that violates liveness property $\Ge$, task $\{t_2\}$ is perpetually
enabled, yet never occurs. Hence, when assuming $\X\Y$-fairness with $\X \in \{\W,\St\}$ and
$\Y\in\{\In,\Sy,\Cp,\Gr\}$, this path is unfair, and the liveness property $\Ge$ is satisfied.
The same holds when assuming $\X\Ts$-fairness, even though this notion gives rise to three tasks.
However, under the fairness assumption ${\St\Ac}$ (and thus certainly under WA) property $\Ge$ is not
satisfied, as all three transitions form one task.
\par}
\end{exam}

\vspace{-1.4mm}
\begin{exam}\label{ex:ST}
Consider \hfill  the \hfill  process \hfill  $a | X$ \hfill  where \hfill  $X \defis b_0.(X[f])$. \hfill  Here \hfill  $f$ \hfill  is \hfill  a \hfill  relabelling \hfill operator
{\makeatletter
\let\par\@@par
\par\parshape0
\everypar{}\begin{wrapfigure}[4]{r}{0.41\textwidth}
 \vspace{-2ex}
  \input{ST}
  \centerline{\raisebox{1ex}{\box\graph}}
  \end{wrapfigure}
\noindent with $f(b_i)=b_{i+1}$ for $i\geq 0$ and $f(a)=a$.
When taking $\Y\in\{A,I,\Sy,C,G\}$ there is a task $\{t_0,t_2,t_4,\dots\}$, which is perpetually enabled
until it occurs. So, under $\X\Y$-fairness
the process does satisfy $\Ge$. Yet under $\X\Ts$-fairness each transition is a separate task, and 
$\Ge$ is not satisfied.
\par
}
\end{exam}

\begin{exam}\label{ex:SG}
\hspace{-.1mm}Consider the process $X$ where $X \defis a.X+b$. Assuming $\X\Y$-fairness with  $\Y{\in}\{\Ac,\Ts,\In,\Sy\}$
{\makeatletter
\let\par\@@par
\par\parshape0
\everypar{}\begin{wrapfigure}[3]{r}{0.4\textwidth}
 \vspace{-2.3ex}
  \input{SG}
  \centerline{\raisebox{1ex}{\box\graph}}
  \end{wrapfigure}
\noindent there is a task
 $\{t_2\}$, so that $\Ge$ is satisfied.
Yet, since there is only one parallel component, $\St\Gr$- and $\St\Cp$-fairness allow an infinite path without
transition $t_2$, so that $\Ge$ is not satisfied.
\par}
\end{exam}

\begin{exam}\label{ex:SC} 
Consider the process $(X|Y)\backslash b$ where $X \defis a.X+b$\vspace{1pt}
{\makeatletter
\let\par\@@par
\par\parshape0
\everypar{}\begin{wrapfigure}[4]{r}{0.3\textwidth}
 \vspace{-5.2ex}
  \input{SC}
  \centerline{\raisebox{1ex}{\box\graph}}
  \end{wrapfigure}
\noindent and \plat{$Y \defis c.Y+\bar b$}. The instructions 
are the action occurrences in the specification: $\I\mathbin=\{a,b,c,\bar b\}$.
We have $\instr(t_1)\mathbin=\{a\}$, $\instr(t_2)=\{b,\bar b\}$ and $\instr(t_3)=\{c\}$.
There are three components, namely 
$X$ (called $\Left$), $Y$ (called $\R$), and the entire expression$(X|Y)\backslash b$, with
$\comp(t_1)=\{\Left\}$, $\comp(t_2)=\{\Left,\R\}$ and $\comp(t_3)=\{\R\}$.
So under fairness of synchronisations, or groups of components, there is a task $\{t_2\}$, making a
path that never performs $t_2$ unfair. The same holds under fairness of actions or transitions, as
well as instructions. Hence $\Ge$ is satisfied. However, under fairness of components, the tasks are
$\{t_1,t_2\}$ and $\{t_2,t_3\}$ and $\Ge$ is not satisfied.
\par}
\end{exam}

In the remainder of this section we establish a hierarchy of fairness properties; 
for this we introduce a partial order on fairness assumptions.

\begin{definition}
Fairness assumption F is \emph{stronger} than fairness assumption H, denoted by ${\rm H} \preceq {\rm F}$,
iff {\rm F} rules out at least all paths that are ruled out by H.
\end{definition}

\noindent Figure~\ref{taxonomy} classifies the above progress and fairness assumptions by strength, with
arrows pointing towards the weaker assumptions. P stands for progress. Arrows derivable by reflexivity or
\begin{wrapfigure}[15]{r}{0.44\textwidth}
\input{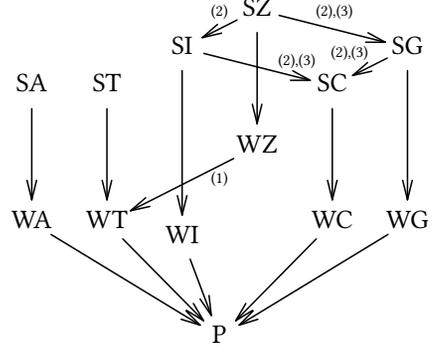}
\centerline{\raisebox{1ex}{\box\graph}}
\vspace{-2pt}
\caption{A classification of progress and fairness}
\label{taxonomy}
\end{wrapfigure}
transitivity of $\preceq$ are suppressed.
The numbered arrows are valid only under the following assumptions:
\begin{enumerate}[(1)]
\item For each synchronisation $Z\mathbin\subseteq \I$, and for each state $s$,
  there is at most one transition $t$ with $\instr(t)\mathbin=Z$ enabled in~$s$.%
  \label{unique synchronisation}
\item $\I$ is finite.
  \label{finite}
\item There is a function $\cmp\!:\I\rightarrow\Ce$ such that
  $\comp(t)\mathbin=\{\cmp(I)\mathbin{\mid} I\mathop\in \instr(t)\}$ for all $t\mathop\in\Tr$.
  \label{cmp}
\end{enumerate}
These assumptions hold for the transition systems derived from the fragment of CCS presented in
\App{fragment}. All other arrows have full generality. The validity of the arrows is shown below
(Props.~\ref{pr:P<W<S}--\ref{pr:SC-SI}).
The absence of any further arrows follows from \exs~\ref{ex:mutex}--\ref{ex:WC}:
  Example~\ref{ex:mutex} separates weak from strong fairness; 
  it shows that $\St\Y \not\preceq \W\Y'$ for all $\Y,\Y' \in \{\Ac,\Ts,\In,\Sy,\Cp,\Gr\}$.
  Example~\ref{ex:SA} shows that there are no further arrows from SA, and thus none from WA\@.
  Example~\ref{ex:ST} shows that there are no further arrows from ST, and thus none from WT.
  Together with transitivity, these two conclusions also imply that
  there are  no further arrows from P\@.
  Examples~\ref{ex:SG},~\ref{ex:SC} and the forthcoming \ex{SI} show that there are no
  further arrows from \St\Gr, SC and SI, respectively, and thus none from WC\@.
Examples~\ref{ex:SS-ST},~\ref{ex:SS-WA} and~\ref{ex:WC} show that there are no further
  arrows into ST, WA and WC, respectively, and thus, using transitivity, no further from S\Sy\@.
  Example~\ref{ex:WS} shows that there are no further arrows from W\Sy.
  Further arrows from WI or WG are already ruled out by transitivity of~$\preceq$.\qed
\pagebreak[4] %% hard pagebreak

\begin{exam}\label{ex:SI}
Consider the process $X | Y$ where $X \defis a.X$ and\vspace{1pt}
{\makeatletter
\let\par\@@par
\par\parshape0
\everypar{}\begin{wrapfigure}[3]{r}{0.34\textwidth}
  \vspace{-6ex}
  \input{SI}
  \centerline{\raisebox{1ex}{\box\graph}}
  \end{wrapfigure}
\noindent \plat{$Y \defis \bar a.Y$}. The infinite path $(a \bar a)^{\infty}$ is $\X\In$-fair and
$\X\Cp$-fair, but not $\X\Y$-fair for $\Y\mathbin\in\{\Ac,\Ts,\Sy,\Gr\}$.
% Another example: (X|Y)\a with X=aX+abX+cX and Y=\bar a.Y + \bar{a}.\bar{c}.Y +\bar{b}.Y.
% Here SZ ensures halting, but SI, SG and SC do not.
\par}
\end{exam}

\begin{exam}\label{ex:SS-ST}
% An alternative is Example 4.30 S' in Francez 1985.
Consider\hfill  the\hfill  process\hfill  $X | Y$\hfill  where\hfill  $X \defis a.X$\hfill  and\hfill  $Y \defis b.c.Y$.\hfill Under\hfill  ST-fairness,\hfill the 
{\makeatletter
\let\par\@@par
\par\parshape0
\everypar{}\begin{wrapfigure}[5]{r}{0.44\textwidth}
 \vspace{-2ex}
  \input{SS-ST}
  \centerline{\raisebox{1ex}{\box\graph}}
  \end{wrapfigure}
\noindent path $(a b c)^\infty$
 is unfair, because task $\{t_3\}$ is infinitely often
enabled, yet never taken. However, under $\X\Y$-fairness for $\Y \in\{\Ac,\In,\Sy,\Cp,\Gr\}$ there is
no task $\{t_3\}$ and the path is fair.
\par}
\end{exam}

\begin{exam}\label{ex:SS-WA}
Consider\hfill the\hfill process\hfill $X$\hfill where\hfill $X \defis a.X+c.X+b.(X[f])$,\hfill using\hfill a\hfill relabelling\hfill oper-
{\makeatletter
\let\par\@@par
\par\parshape0
\everypar{}\begin{wrapfigure}[6]{r}{0.44\textwidth}
  \input{SS-WA}
  \centerline{\raisebox{1ex}{\box\graph}}
  \end{wrapfigure}
\noindent ator with $f(a)\mathbin=c$, $f(c)\mathbin=a$ and $f(b)\mathbin=b$. 
On the right, as always, we depict the unique transition
  system generated by the operational semantics of CCS\@. There exists, however, a 2-state
  transition system that is strongly bisimilar with it.
Under $x\Ac$-fairness, the path $(ab)^\infty$ is unfair, because task $c$ is perpetually
enabled, yet never taken. However, there are only three instructions (and one component), each of
which gets its fair share. So under $\X\Y$-fairness for $\Y \in\{\In,\Sy,\Cp,\Gr\}$ the path is fair.
It is also $\X\Ts$ fair, as no transition is perpetually or relentlessly enabled.%
\footnote{In the transition system consisting of 2 states the path is $\St\Ts$ fair; see also \Sec{unfolding}.}
\par}
\end{exam}

\begin{exam}\label{ex:WS} 
Consider the process $(a | X | Y)\backslash a$ where $X \defis \bar a.\bar a.X$ and\vspace{1pt}
{\makeatletter
\let\par\@@par
\par\parshape0
\everypar{}\begin{wrapfigure}[5]{r}{0.25\textwidth}
 \vspace{-6ex}
  \input{WS}
  \centerline{\raisebox{1ex}{\box\graph}}
   \vspace{-10ex}
  \end{wrapfigure}
\noindent  \plat{$Y \defis a.Y$}. 
Under $\W\Y$-fairness with $\Y\mathbin\in\{\In,\Cp,\Gr\}$ property $\Ge$ is satisfied, but under
$\W\Sy$- and $\W\Ts$-fairness it is not,
because each synchronisation is enabled merely in every other state, not perpetually.
Since there is only a single action $\tau$, 
under $\X\Ac$-fairness property $\Ge$ is not satisfied.
\par}
\end{exam}

\begin{exam}\label{ex:WC}
Consider the process $((a{+}c) | X | Y | Z)\backslash\hspace{-1pt} a\backslash\hspace{-1pt} b\backslash\hspace{-1pt} c\backslash\hspace{-1pt} d$\vspace{1pt}

{\makeatletter
\let\par\@@par
\par\parshape0
\everypar{}\begin{wrapfigure}[2]{r}{0.36\textwidth}
 \vspace{-6.2ex}
  \input{WC}
  \centerline{\raisebox{1ex}{\box\graph}}
  \end{wrapfigure}
\noindent
where \plat{$X \defis \bar a.\bar b.X$}, \plat{$Y \defis \bar c.\bar d.Y$} and \plat{$Z \defis a.b.c.d.Z$}.
Under $\W\Cp$-fairness $\Ge$ is satisfied, but under $\W\Y$-fairness with $\Y\mathbin\in\{\Ac,\Ts,\In,\Sy,\Gr\}$ it is not.
\par}
\end{exam}

\begin{proposition}\rm\label{pr:P<W<S}
${\rm P} \preceq \W\Y \preceq \St\Y$ for $\Y \in \{\Ac,\Ts,\In,\Sy,\Cp,\Gr\}$.
\end{proposition}
\begin{minipage}{.6\textwidth}
\begin{proof}
That $\W\Y \preceq \St\Y$ follows from \df{fair}.
That \linebreak[3]\hspace*{-10pt}%
${\rm P} \preceq \W\Y$ follows since for each $\Y$ we have $\bigcup_{T\in\T}T=\Tr$.
\end{proof}
\end{minipage}

\begin{proposition}\rm
When assuming \ref{unique synchronisation}, $\W\Ts \preceq \W\Sy$.
\end{proposition}

\begin{proof}
\newcommand{\pwt}{\mathfrak{p}_{\scalebox{.5}{\W\Ts}}}
We show that any $\W\Sy$-fair path is also $\W\Ts$-fair.
Let $\pwt(\pi)$ be the property of a path $\pi$ that ``each $\W\Ts$-task that is perpetually enabled on $\pi$, occurs in $\pi$.''
  By \df{fair} we have to show that ``if a path $\pi$ is $\W\Sy$-fair, then each suffix $\pi'$ of $\pi$ satisfies $\pwt(\pi')$''.
  To this end it suffices to show that ``if $\pi$ is $\W\Sy$-fair, then $\pwt(\pi)$''.
  Namely, if $\pi$ is $\W\Sy$-fair, then each suffix $\pi'$ of $\pi$ is $\W\Sy$-fair, and thus
  satisfies $\pwt(\pi')$.

  So let $\pi$ be a {$\W\Sy$}-fair path. Each $\W\Ts$-task has the form $\{t\}$, and is
  perpetually enabled on $\pi$ iff $t$ is. Assume that a transition $t$ is
  perpetually enabled on $\pi$. We need to show that $t$ occurs in $\pi$.

Since $t$ is perpetually enabled on $\pi$, $T_{\instr(t)}$, as defined on Page~\pageref{sy_synchronisation},
is a W\Sy-task that is perpetually enabled on $\pi$.
So a transition in $T_{\instr(t)}$ must eventually occur, say in state $s$ of $\pi$.
By \ref{unique synchronisation} $t$ is the only transition from $T_{\instr(t)}$ enabled in $s$.
So $t$ will occur in $\pi$.
\end{proof}

\noindent
In future proofs, the reduction from proving a property for each suffix $\pi'$ of $\pi$ to
  simply proving the property for $\pi$, spelled out in the first paragraph of the above proof,
  will be taken for granted.

\begin{proposition}\rm\label{pr:SS-SI}
If  $\I$ is finite (Property \ref{finite}) then $\St\In \preceq \St\Sy$.
\end{proposition}
\begin{proof}
  We show that any $\St\Sy$-fair path is also $\St\In$-fair.
  By \ref{finite} there are only finitely many SI- and S\Sy-tasks.
  Moreover, each SI-task is the union of a collection of S\Sy-tasks.
  Let $\pi$ be a S\Sy-fair path on which an SI-task $T$ is infinitely often enabled (or in the last state).
  Then an S\Sy-task $T'\subseteq T$ must be be infinitely often enabled on $\pi$ (or in the last state).
  So $T'$, and hence $T$, will occur in $\pi$.
\end{proof}

\begin{proposition}\rm\label{pr:SC-SI}
If $\I$ is finite (Property \ref{finite}) and 
there is a function $\cmp\!:\I\rightarrow\Ce$ such that
  $\comp(t)\mathbin=\{\cmp(I)\mathbin{\mid} I\mathop\in \instr(t)\}$ for all $t\mathop\in\Tr$
  (Property \ref{cmp}) 
then $\St\Cp \preceq \St\In$ and $\St\Cp \preceq \St\Gr \preceq \St\Sy$.
\end{proposition}
\begin{proof}
  By \ref{finite} and \ref{cmp} there are only finitely many SI-, S\Sy-, SC- and SG-tasks.
  Moreover, each SC-task is the union of a collection of SI-tasks;
  each SG-task is the union of a collection of S\Sy-tasks;
  and each SC-task is the union of a collection of SG-tasks.
  The rest of the proof proceeds as for \pr{SS-SI}.
\end{proof}

\section{Criteria for Evaluating Notions of Fairness}\label{sec:criteria}

In this section we review four criteria for evaluating fairness properties---the first three taken
from \cite{AFK88}---and evaluate the notions of fairness presented so far.
We treat the first criterion---feasibility---as prescriptive, in that we discard from further
consideration any notion of fairness not meeting this criterion. The other criteria---equivalence
robustness, liveness enhancement, and preservation under unfolding---are merely 
arguments in selecting a particular notion of fairness for an application.
In \Sec{J-fairness} we propose a fifth criterion---\emph{preservation under refinement of
actions}---and argue that none of the weak notions of fairness satisfies that latter criterion.

\subsection{Feasibility}\label{sec:feasibility}
The purpose of fairness properties is to reduce the set of infinite potential runs of systems,
without altering the set of finite partial runs in any way. Hence a natural requirement 
is that any finite partial run can be extended to a fair run. This appraisal criteria
on fairness properties has been proposed by {\sc Apt, Francez \& Katz} in \cite{AFK88} and called
\emph{feasibility}. It also appears in {\sc Lamport}~\cite{La00} under the name \emph{machine closure}.
We agree with \textsc{Apt et al.} and \textsc{Lamport} that this is a necessity for any sensible
notion of fairness. By means of the following theorem we show that this criterion is satisfied by all
fairness properties reviewed so far, when applied to the fragment of CCS from \App{fragment}.

\begin{theorem}\rm\label{thm:feasibility}
  If, in a transition system, only countably many tasks are enabled in each state
then the resulting notions of weak and strong fairness are feasible.
\end{theorem}
\begin{proof}
  We present an algorithm for extending any given finite path $\pi_0$ into a fair path $\pi$.
  We build an $\IN\times\IN$-matrix with a column for the---to be constructed---prefixes $\pi_{i}$
  of $\pi$, for $i\geq 0$.
  The columns $\pi_i$ will list the tasks enabled in the last state of $\pi_i$, leaving empty most
  slots if there are only finitely many.
  An entry in the matrix is either (still) empty, filled in with a task, or crossed out.
  Let $f:\IN\rightarrow \IN\times\IN$ be an enumeration of the entries in this matrix.
  
  At the beginning only $\pi_0$ is known, and all columns of the matrix are empty.
  At each step $i\geq 0$ we proceed as follows:
  
  Since $\pi_i$ is known, we fill the $i$-th column by listing all enabled tasks.
  In case no task is enabled in the last state of $\pi_i$, the algorithm terminates, with output $\pi_i$.
  Otherwise, we take $j$ to be the smallest value such that entry $f(j)\in\IN\times\IN$ is already
  filled in, but not yet crossed out, and the task $T$ listed at $f(j)$ also occurs in
  column $i$ (\ie is enabled in the last state of $\pi_i$). We now extend $\pi_i$ into $\pi_{i+1}$
  by appending a transition from $T$ to it, while crossing out entry $f(j)$. 
  
  Obviously, $\pi_{i}$ is a prefix of $\pi_{i+1}$, for $i\geq 0$.
  The desired fair path $\pi$ is the limit of all the $\pi_i$.
  It is strongly fair (and thus weakly fair), because each task that is even once enabled
  will appear in the matrix, which acts like a priority queue. Since there are only finitely many
  tasks with a higher priority, at some point this task will be scheduled as soon as it occurs again.
\end{proof}

Since our fragment of CCS yields only finitely many instructions and components, there are only
finitely many tasks according to the fairness notions $\In,\Sy,\Cp$ and $\Gr$.
Moreover, there are only
countably many transitions in the semantics of this fragment, and consequently only countably many
action labels. Hence all fairness notions reviewed so far, applied to this fragment of CCS, are feasible.

\subsection{Equivalence robustness}\label{sec:robustness}
\begin{exam}\label{ex:robustness}
Consider the process $(X | Y | Z)\backslash e$ where \vspace{1pt}
{\makeatletter
\let\par\@@par
\par\parshape0
\everypar{}\begin{wrapfigure}[10]{r}{0.41\textwidth}
 \vspace{-6.5ex}
  \input{ER}
  \centerline{\raisebox{1ex}{\box\graph}\quad}
  \end{wrapfigure}
\noindent
\plat{$X \defis a.b.X + e$}, \plat{$Y \defis \bar e$} and \plat{$Z \defis c.d.Z + e$}.
Under fairness of components, the four $\tau$-transitions form one task, as these are the transitions
that synchronise with component $Y$. Thus, under WC-fairness, the path $(a b c d)^\infty$ is unfair, because
that task is continuously enabled, yet never taken. However, the path $(a c b d)^\infty$ is fair,
since this task is infinitely often not enabled. In \emph{partial order semantics} \cite{Pr86}
these two runs are considered \emph{equivalent}: both model the parallel composition of
runs $(a b)^\infty$ of component $X$ and $(c d)^\infty$ of component $Z$, with no causal dependencies
between any actions  occurring in the latter two runs.
For this reason WC-fairness is is not robust under equivalence, as proposed in \cite{AFK88}.
\par}
\end{exam}

In general, {\sc Apt, Francez \& Katz} call a fairness notion \emph{equivalence robust} if
``for two infinite sequences which differ by a possibly infinite number of interchanges
of independent actions (\ie\ equivalent sequences), either both are fair according to the given
definition, or both are unfair.'' \cite{AFK88}

\begin{exam}\label{ex:robustness strong}
{\newcommand{\YP}{Z}
\hspace{2pt}Consider\hspace{2pt} the\hspace{2pt} process\hspace{2pt} $(X | \YP)\backslash e$\hspace{2pt} where\vspace{1pt}
{\makeatletter
\let\par\@@par
\par\parshape0
\everypar{}\begin{wrapfigure}[8]{r}{0.41\textwidth}
 \vspace{-4ex}
  \input{ERS}
  \centerline{\mbox{}\hspace{5.5pt}\raisebox{1ex}{\box\graph}\quad}
  \end{wrapfigure}
\noindent
\plat{$X \defis a.b.X + e$} and \plat{$\YP \defis c.d.\YP + \bar e$}.
Under fairness of actions, transitions, instructions, synchronisations or groups of components,
the $\tau$-transition forms a separate task. Thus, under $S\Y$-fairness, for
$\Y \in\{\Ac,\Ts,\In,\Sy,\Gr\}$, the path $(a c b d)^\infty$ is unfair, because
that task is infinitely often enabled, yet never taken. However, the path $a (c b a d)^\infty$ is fair,
since from its second state onwards this task is never enabled. Yet these two runs are partial-order equivalent, so
also $S\Y$-fairness for $\Y \in\{\Ac,\Ts,\In,\Sy,\Gr\}$ is not equivalence robust.
\par}
}%end\newcommand{\Yp}
\end{exam}

That WC-, SG- and SZ-fairness are not equivalence robust was established by {\sc Apt et al.\ } \cite{AFK88}.
There it is further shown that on a fragment of the process algebra CSP \cite{Ho78}, featuring only
2-way synchronisation, SC-fairness is equivalence robust. We do not know if this result extends to our
fragment of CCS, as it is in some ways less restrictive than the employed fragment of CSP\@.
In the presence of $N$-way synchronisation (with $N\mathord>2$) SC-fairness is not equivalence robust  \cite{AFK88};
this is shown by a variant of Example~\ref{ex:robustness strong} that features an extra component
$Z$, performing a single action that synchronises with both $e$ and $\bar e$.
{\sc Apt et al.}\ also show that, on their fragment of CSP, WG- and WZ-fairness are  equivalence robust.
We do not know if these results hold in full generality, or for our fragment of CCS\@.
Example~\ref{ex:robustness} shows that WA- and WI-fairness are not equivalence robust.

\subsection{Liveness enhancement}\label{sec:enhancement}

A notion of fairness F is \emph{liveness enhancing} if there exists a liveness property that holds
for some system when assuming F, but not without assuming F\@. This criterion for appraising fairness notions
stems from \cite{AFK88}, focusing on the liveness property of termination; there,  F is liveness
enhancing iff there exists a system such that all its fair runs terminate, but not all its unfair runs.

When taken literary as worded above, the assumption of progress is liveness enhancing, since the
program from Example \ref{ex:progress} does satisfy the liveness property $\Ge$ when assuming
progress, but not without assuming progress. Since all fairness properties reviewed here
are stronger then progress, it follows that they are liveness enhancing as well.

The notion becomes more interesting in the presence of a background assumption, which is weaker than
the fairness assumption being appraised. Assumption F then becomes liveness enhancing if there exist
a liveness property that holds for some system when assuming F but not when merely making the
background assumption. In \cite{AFK88} \emph{minimal progress} is used as background assumption: `Every
process in a state with enabled \emph{local} actions will eventually execute some action.'
Here a process is what we call a component, and an action of that
component is \emph{local} if it it is not a synchronisation.
This is a weaker form of the assumption of \emph{justness}, to be introduced in
Section~\ref{sec:justness}. 

In \cite{AFK88} it is established that WC-, WG- and WZ-fairness are not liveness enhancing on the employed
fragment of CSP, whereas SC, SG- and SZ-fairness are. The negative results are due to the limited
expressive power of that fragment of CSP; in our setting all notions of fairness reviewed so far
are liveness enhancing w.r.t.\ the background assumption of minimal progress. This follows from
Examples \ref{ex:SG} and \ref{ex:WS}.

\subsection{Preservation under unfolding}\label{sec:unfolding}

A possible criterion on fairness notions is that it ought to be possible to unfold a transition
system $G$ into a tree---which yields a bijective correspondence between the paths in $G$ and those
in its unfolding---without changing the set of fair paths.

This criterion holds for all notions of fairness considered so far, except for fairness of transitions.
The reason is that any transition in the unfolding is inherited from a transition $t$ in $G$,
and thereby inherits properties like $\ell(t)$, $\instr(t)$ and $\comp(t)$.

When using fairness of transitions
in the first transition system of \ex{weak} the task $\T$ consists of two elements, and
$\Ge$ is ensured, whereas the second transition system yields an infinite set $\T$, and $\Ge$ is
not guaranteed, even though the second transition system is simply an unfolding of the first.
\ex{SS-WA} illustrates another scenario where a path is not $\St\Ts$-fair in a finite transition system, 
but $\St\Ts$-fair in its unfolding.
This could be regarded as a drawback of fairness of transitions.

In \cite{QS83}, where fairness of transitions is employed, care is taken to select a particular
transition system for any given program. The one chosen has exactly one state for each combination
of control locations, and values of variables. So in terms of \ex{weak} it is the variant of the
second transition system in which all final states are collapsed.

\section{Strong Weak Fairness}\label{sec:sw}

We present an example of a system for which none of the above fairness notions 
is appropriate: weak fairness is too weak and strong fairness too strong. We consider this
example not to be a corner case, but rather an illustration of the typical way fairness is
applied to real systems. Consequently, this example places serious doubts on the notions of fairness
reviewed so far. Failing to find a concept of fairness in the literature that completely fits the requirements of
this example in the literature, we propose the concept of \emph{strong weak fairness}, which is the logical
combination of two of the original fairness concepts contemplated in the literature.

\begin{exam}\label{ex:3q}
Let $X$ be a clerk that has to serve customers that queue
at three windows of an office. Serving a customer is a multi-step process.\vspace{1pt} 
During such a process the clerk is not able to serve any other customer.
A CCS specification of the clerk's behaviour
is \plat{$X\defis c_1.e.X + c_2.e.X + c_3.e.X$}. Here $c_i$, for
$i=1,2,3$, is the action of starting to serve a customer waiting at
window $i$, and $e$ the remainder of the serve effort.
If this were a section on reactive systems we could continue the
example by evaluating this clerk in any possible context.
But since we promised (at the end of the introduction) to deal, for
\vspace{1pt}
the time being, with closed systems only, we close the system by
specifying the three rows of customers.
Let \plat{$Y_1 \defis \bar c_1.Y_1$} model an never ending sequence of customers queuing at window $1$, with $\bar c_i$ being the action of being served at
window $i$. 
Likewise \plat{$Y_2 \defis \bar c_2.Y_2$}, but at window $3$ we place
only two customers that may take the action $g$ of going home when not
being served for a long time, and the action $r$ of returning to window 3 later on:
\[
\begin{array}{c@{~\defis~}l@{\qquad}c@{~\defis~}l@{\qquad}c@{~\defis~}l}
Y^{2,2}_3& \bar c_3.Y^{1,1}_3 + g.Y^{2,1}_3 &
Y^{2,1}_3& \bar c_3.Y^{1,0}_3 + g.Y^{2,0}_3 + r.Y^{2,2}_{3} &
Y^{2,0}_3& r.Y^{2,1}_{3}\\
Y^{1,1}_3& \bar c_3.Y^{0,0}_3 + g.Y^{1,0}_3 &
Y^{1,0}_3& r.Y^{1,1}_3 &
Y^{0,0}_3& {\bf 0} 
\end{array}
\]
Here $Y^{i,j}_3$ is the state in which there are $i$ potential customers for window 3, of which $j$
are currently queueing.
Our progress requirement says that if both customers have gone home, eventually one of them will
return to window~3.
The overall system is $(X | Y_1 | Y_2 | Y^{2,2}_3)\backslash c_1\backslash c_2\backslash c_3$.
An interesting liveness property $\Ge$ is that any customer, including the ones
waiting at window 3, eventually gets served. Two weaker properties are $\Ge_2$: eventually a
customer waiting at window~2 will get served, and $\Ge_3$: eventually either a customer waiting at
window~3 will get served, or both customers that were waiting at window 3 are at home.

As manager of the office, we stipulate that the clerk should treat all
windows fairly, without imposing a particular scheduling strategy.
Which fairness property do we really want to impose?

We can exclude fairness of actions here, because, due to synchronisation into $\tau$, no action
differentiates between serving different customers:
neither $\W\Ac$ nor $\St\Ac$ does ensure $\Ge_2$.

None of the weak fairness properties presented so far is strong enough.
For suppose that the clerk would only ever serve window 1.
Then the customers waiting at windows 2 and 3 represent tasks that
will never be scheduled. However, these tasks are not
perpetually enabled, because each time the clerk is between actions
$c_1$ and $e$, they are not enabled. For this reason, no weak fairness
property is violated in this scenario.
In short: $\W\Y$ for $\Y \in \{\Ts,\In,\Sy,\Cp,\Gr\}$ does not even ensure $\Ge_2$.

Thus it appears that as manager we have to ask the clerk to be
strongly fair: each component (window) that is relentlessly enabled
(by having a customer waiting there when the clerk is ready to serve a new customer), should eventually be served.
In particular, this fairness requirement (FR) does not allow skipping
window 3 altogether, even if its last customer goes home and returns
infinitely often.\footnote{However, it does allow skipping window 3 altogether if both customers are mostly at home,
and can be found at the window only when the clerk is active serving another customer.\label{scenario}}
Technically, FR is implied by imposing $\St\Y$ for some $\Y \in \{\Ts,\In,\Sy,\Gr\}$. However,
{\St\Cp} is too weak; it does not even ensure $\Ge_3$, because when the queue-length at window~3
alternates between 2 and 1, that component is deemed to get its fair share.

Suppose that the clerk implements FR by a
\emph{round robin} scheduling strategy. The clerk first
tries to serve a customer from window 1.
If there is no customer at window 1, or after interacting with a client from 
window 1,  she serves a customer from window 2 if there is one, and so on,
returning to window 1 after dealing with window 3.

At first sight, this appears to be an entirely fair scheduling strategy.
However, it does not satisfy the strong fairness requirement FR\@: 
suppose that the customers from window 3 spend most of their time at home,
with only one of them showing up at window 3 for a short time every day,
and always timed as unfortunate as possible, \ie arriving right after the clerk
started to serve window~1, and departing right before she finishes
serving window 2.\footnote{Thus avoiding the scenario of Footnote~\ref{scenario}.} In that case these customers will never be served, and
hence FR is not met.
Waiting at each window for the next customer to arrive does not work either as the clerk 
will starve after 2 rounds at window 3.

One possible reaction to this conclusion is to reject round robin scheduling;
an alternative is the ``queueing'' algorithm of {\sc Park}~\cite[Sect.~3.3.2]{Pa81}:
``at each stage, the earliest clause [window] with true guard [a waiting customer]
is obeyed [served], and moved to the end of the queue.''
However, another quite reasonable reaction is that round-robin
scheduling is good enough, and that the waiting/going home scheduling
of customer 3 contemplated above is so restrictive that as manager we
will not be concerned about customer 3 never being scheduled.
This customer has only himself to blame. This reaction calls for a
fairness requirement that is stronger than weak fairness but weaker
than FR, in that it ensures liveness properties $\Ge_2$ and $\Ge_3$, but not $\Ge$.
\end{exam}

{\sc Park}~\cite{Pa81} essentially rejects strong fairness because of the complexity of its
implementation: ``The problem of implementing strong fairness is disquieting.
It is not clear that there is any algorithm which is essentially more efficient
than the queueing algorithm of 3.3.2. [\dots] If the problem is essentially
as complex as this, then strong fairness in this form would seem an undesirable
ingredient of language specification''. In \cite{FP83},\ {\sc Fischer \& Paterson} confirm the state of
affairs as feared by {\sc Park}: any scheduling algorithm for the clerk above requires at least $n!$
storage states---where $n$ is the number of windows---whereas a weakly fair scheduling algorithm
requires only $n$ storage states.

\subsection{Response to insistence and persistence}

One of the very first formalisations of fairness found in the literature occur in \cite{GPSS80},
where two notions of fairness were proposed.
\emph{Response to insistence} 
``states that a permanent holding of a
condition or a request $p$ will eventually force a response $q$.''
\emph{Response to persistence} states ``that the infinitely repeating
occurrence of the condition $p$ will eventually cause $q$.''

The later notions of weak and strong fairness \cite{AO83,GFMdR85,LPS81}, reviewed in
Sections~\ref{sec:fairness} and~\ref{sec:taxonomy}, can be seen as instantiations of the
notions from \cite{GPSS80}, namely by taking the condition $p$ to mean that a task is enabled, and
$q$ that it actually occurs. However, the fairness notions of \cite{GPSS80} also allow different
instantiations, more compatible with \ex{3q}.
We can take $p_i$ to be the condition that a customer is waiting at window $i$, and $q_i$ the
response that such a customer is served. The resulting notion of response to insistence
is much stronger than weak fairness, for it disallows the scenario were the clerk only serves
window~1. In fact, it ensures liveness properties $\Ge_2$ and $\Ge_3$.
Yet, it is correctly implemented by a round robin scheduling strategy, and
does not guarantee the unrealistic liveness property $\Ge$.
Qua complexity of its implementation it sides with weak fairness, as there is no need for more than
$n$ storage states; see the last paragraph of the previous section.

In \cite{TR13} we use fairness properties for the formalisation of temporal properties of routing
protocols. The situation there is fairly similar to \ex{3q}. The fairness properties ($P_1$), ($P_2$)
and ($P_3$) that we propose there are local fairness properties rather than global ones (see
\Sec{taxonomy}) and instances of response to insistence.

Response to insistence
may fail the criterion of feasibility.
Without loss of generality, we can identify the response $q$ of
response to insistence with the task $T$ of \df{fair}.
Response to insistence says that the permanent holding of the request $p$ is enough to ensure
$q$, even if $q$ is, from some point onwards, never enabled.
This can happen in a variant of \ex{3q} where the clerk is given the possibility of serving a single
customer for an unbounded amount of time.
Such a scenario clearly runs contrary to the intentions of \cite{GPSS80}.
A more useful fairness notion, probably closer to the intentions of \cite{GPSS80}, is
\begin{quote}\it
  if a condition or request $p$ holds perpetually from some point onwards, and a response $q$ is
  infinitely often enabled, then $q$ must occur.
\end{quote}
This is the logical combination of strong fairness and response to insistence.
We propose to call it \emph{strong weak fairness}.

\newcommand{\ri}{response to insistence}

Footnote 60 in \cite{TR13} shows that the instances of {\ri} used in that paper
do meet the criterion of feasibility, and hence amount to instances of strong weak fairness.

For the same reason as above, response to persistence does not meet the criterion of
feasibility either. To make it more useful, one can combine it with strong fairness in the same
fashion as above. However, the resulting notion does not offer much that is not offered by
strong fairness alone.

\subsection{Strong weak fairness of instructions}\label{sec:swi}

We finish this section by isolating a particular form for strong weak fairness by choosing the
relevant conditions $p$ and responses $q$. This is similar to the instantiation of the notions of
strong and weak fairness from \Sec{fairness} by choosing the relevant set of tasks in \Sec{taxonomy}.
Here we only deal with \emph{strong weak fairness of instructions} (SWI), by following the choices made in 
\Sec{taxonomy} for strong and weak fairness of instructions.

To provide the right setting, we make the following two assumptions on our transition system, all of
which are satisfied on the studied fragment of CCS\@.
By property \ref{cmp} in \Sec{taxonomy}, each instruction~$I$ belongs to a component $\cmp(I)$.
Given a state $P$ of the overall system, a given component $C \in \Ce$, if active at all, is in a
state corresponding to a subexpression of $P$, denoted by $P_C$.

We say that an instruction $I$ is \emph{requested} in state $P$ if a transition $t$ with
$I\in\instr(t)$ is enabled in $P_{\cmp(I)}$. The instruction is \emph{enabled} in $P$ if a transition $t$
with $I\in\instr(t)$ is enabled in $P$. The instruction \emph{occurs} in a path $\pi$ if $\pi$ contains
a transition $t$ with $I\in\instr(t)$. We recall that a path $\pi$ is strongly fair according to
fairness of instructions if, for every suffix $\pi'$ of $\pi$, each task that is relentlessly
enabled on $\pi'$, occurs in $\pi'$.
\begin{definition}
  A path $\pi$ in a transition system is \emph{strongly weakly fair} if, for every suffix $\pi'$ of $\pi$,
  each instruction that is perpetually requested and relentlessly enabled on $\pi'$, occurs in $\pi'$.
\end{definition}
Applied to \ex{3q}, instruction $\bar c_2$ is requested in each state of the
system, and enabled in each state where none of the instructions $e$ is enabled.
Instruction $\bar c_2$ being requested signifies that a customer is waiting at window~2.
Thus, on any rooted path $\pi$ instruction $\bar c_2$ is perpetually requested and relentlessly enabled,
even though it is not perpetually enabled.
Consequently, by strong weak fairness of instructions, but not by weak fairness of instructions,
a synchronisation involving $\bar c_2$ will eventually occur, meaning that window~2 is served.
So $\Ge_2$ is ensured.
Likewise, $\Ge_3$ is ensured, but $\Ge$ is not, for the moment that both customers of
window~3 are at home, the condition of $\bar c_{3}$ being requested is interrupted.

\subsection{Adding strong weak fairness of instructions to our taxonomy}\label{sec:adding}

For our classification of SWI-fairness we require two assumptions:
\begin{enumerate}[(1)]
\setcounter{enumi}{3}
\item\label{swi1} If an instruction $I$ is enabled in a state $P$, it is also requested.
\item\label{swi2} If $I$ is requested in state $P$ and $u$ is a
transition from $P$ to $Q$ such that $\cmp(I) \notin \comp(u)$, then $I$ is still requested in $Q$.
\end{enumerate}
Both \ref{swi1} and \ref{swi2} hold on our fragment of CCS\@.

By \ref{swi1}, SWI is a stronger notion of fairness then {\W\In}. Strictness follows
by \ex{3q}, using $\Ge_2$.  By definition, SWI is weaker than {\St\In}.  Strictness follows by
\ex{3q}, using $\Ge$.  Interestingly, SWI is stronger than {\St\Cp}, even though SWI has the
characteristics of a weak fairness notion.

\begin{proposition}\rm\label{pr:SWI-SC}
When assuming \ref{finite}--\ref{swi2} then $\St\Cp \preceq \St\W\In$.
\end{proposition}
\begin{proof}
  We show that any SWI-fair path is also SC-fair. The proof is by contradiction. 
  Suppose $\pi$  is  $\St\W\In$-fair but  not $\St\Cp$-fair.
  Then there exists a component $C$ such that $T_C$, which was given as $\{t\in\Tr \mid C\in\comp(t)\}$, is
  relentlessly enabled on $\pi$, yet never occurs in $\pi$.
  Since each \St\Cp-task is the union of finitely many \St\In-tasks
  (using \ref{finite} and \ref{cmp}; see \pr{SC-SI}),
  there must be an instruction $I$ with $\cmp(I)=C$ that is relentlessly enabled on $\pi$.
  Let $P$ be a state of $\pi$ in which $I$ is enabled.
  Then, by \ref{swi1}, $I$ is also requested in state $P$.
  But since no transition $t$ with $C\in\comp(t)$ is ever scheduled, by \ref{swi2},
  $I$ remains requested for that component
  in all subsequent states, and hence is perpetually requested. By SWI-fairness, $I$, and thereby $T_C$,
  will occur in $\pi$, contradicting the assumption.
\end{proof}
Hence the position of SWI in our hierarchy of fairness notions is as indicated in \fig{full taxonomy}.
There cannot be any further arrows into or out of SWI, because this would give rise to new arrows
between the other notions.

\section{Linear-time Temporal Logic}\label{sec:ltl}

Fairness properties can be expressed concisely in Linear-time Temporal Logic (LTL).
LTL formulas~\cite{Pn77} are built from a set AP of \emph{atomic propositions} $p,q,\dots$ by means of unary
operators \textbf{F} and \textbf{G}\footnote{In later work operators \textbf{X} and \textbf{U} are
  added to the syntax of LTL; these are not needed here.}\linebreak[2]
and the connectives of propositional logic. 
Informally, a path $\pi$ satisfies $\textbf{F}p$, for atomic proposition $p$, if there is a state of $\pi$ that satisfies $p$, meaning that 
$p$ is reached \emph{eventually}. 
A path $\pi$ satisfies $\textbf{G}p$ if all states of $\pi$ satisfy $p$, meaning that $p$ holds \emph{globally}.

Formally, all operators are interpreted on total transition systems $G=(S, \Tr, \source,\target,I)$ as in \df{TS} that
are equipped with a validity relation ${\models} \subseteq S \times \textrm{AP}$ that tells which
atomic propositions hold in which states. 
Here a transition system is \emph{total} iff each state has an outgoing transition.
It is inductively defined when an infinite path $\pi$ in $G$ \emph{satisfies} an LTL formula
$\varphi$---notation $\pi\models\varphi$:
\begin{itemize}
\item $\pi= s_0 t_1 s_1 t_2 s_2 \dots \models p$ iff $s_0 \models p$;
\item $\pi \models \varphi \wedge \psi$ iff $\pi \models \varphi$ and $\pi \models \psi$;
\item $\pi \models \neg\varphi$ iff $\pi \not\models \varphi$;
\item $\pi \models \textbf{F}\psi$ iff there is a suffix $\pi'$ of $\pi$ with $\pi'\models\phi$;
\item $\pi \models \textbf{G}\psi$ iff for each suffix $\pi'$ of $\pi$ one has $\pi'\models\phi$.
\end{itemize}
The transition system $G$ satisfies $\varphi$---notation $G\models\varphi$---if $\pi\models\varphi$ for each infinite rooted
path in $G$.

This definition of satisfaction has a progress assumption built-in,
for totality implies that a path is infinite iff it is progressing. LTL can be applied to transition systems that are not
necessarily total by simply replacing ``infinite'' by ``progressing'' in the above definition
(cf.~\Sec{progress}).

We apply LTL to transition systems $G$ were some relevant propositions $q$ are defined on the
\emph{transitions} rather then the states of $G$. This requires a conversion of $G$ into a related
transition system $G'$ where those propositions are shifted to the states, and hence can be used as
the atomic propositions of LTL\@. Many suitable conversions appear in the literature \cite{DV95,GV06}.
For our purposes, the following partial unfolding suffices:

\begin{definition}
Given an augmented transition system 
\[G=(S, \Tr, \source,\target,I,\textrm{AP}_S,\models_S,\textrm{AP}_\Tr\/,\models_\Tr)\]
where ${\models_S} \subseteq S \times \textrm{AP}_S$ supplies the validity of state-based
atomic propositions, and ${\models_\Tr} \subseteq \Tr \times \textrm{AP}_\Tr$ supplies the validity
of transition-based atomic propositions.

The converted transition system $G'$ is defined as $(S', \Tr', \source',\target',I,\textrm{AP},\models)$, where
\begin{itemize}
\item$S':=I \cup \Tr$\ ,
\item $\Tr':= \{(\source(t),t)\mid t\mathbin\in\Tr \wedge \source(t)\in I\} \cup \{(t,u)\in\Tr\times\Tr \mid \target(t)=\source(u)\}$\ ,
\item $\source'(t,u):=t$\ ,
\item $\target'(t,u):=u$\ ,
\item $\textrm{AP} := \textrm{AP}_S \uplus \textrm{AP}_\Tr$ , and
\item $t \models p$ iff $p \in \textrm{AP}_S \wedge(
   (t\in I \wedge t\models_S p) \vee (t\in\Tr\wedge \target(t) \models_S p))$\\
   \phantom{$t \models p$ iff }or 
$p \in \textrm{AP}_{\Tr} \wedge t\in\Tr\wedge t \models_{\Tr} p$\ .
\end{itemize}
\end{definition}

This construction simply makes a copy of each state for each transition that enters it, so that a
proposition pertaining to a transition can be shifted to the target of that transition, without
running into ambiguity when a state is the target of multiple transitions.
Instead of calling a state $(t,\target(t))$, we simply denote it by the transition $t$.
We say that $G \models \varphi$ iff $G' \models \varphi$, for $\varphi$ an LTL formula that may
use state-based as well as transition-based atomic propositions.

In \cite{GPSS80}, {\ri} was expressed in LTL as
$\textbf{G}(\textbf{G}p \Rightarrow \textbf{F}q)$.
Here $\textbf{G}p \Rightarrow \textbf{F}q$ ``states that a permanent holding of a
condition or a request $p$ will eventually force a response $q$.''
``Sometimes, the response $q$ frees the requester from being frozen at the requesting state. In this
case once $q$ becomes true, $p$ ceases to hold, apparently falsifying the hypothesis $\textbf{G}p$.
This difficulty is only interpretational and we can write instead the logically equivalent condition
$\neg\textbf{G}(p \wedge \neg q)$'' \cite{GPSS80}.
The outermost $\textbf{G}$ requires this condition to hold for all future behaviours.

Using that $\neg\textbf{G}(p \wedge \neg q)$ is equivalent to $\textbf{F}\neg p \vee \textbf{F} q$,
and $\textbf{G}(\textbf{F}\neg p \vee \textbf{F}q)$ is equivalent to
$\textbf{G}\textbf{F}\neg p \vee \textbf{G}\textbf{F}q$, it follows that 
{\ri} can equivalently be expressed as
$\textbf{F}\textbf{G}p \Rightarrow \textbf{G}\textbf{F}q$, saying that if from some point onwards
condition $p$ holds perpetually, then response $q$ will occur infinitely often.

Weak fairness, as formalised in \Sec{fairness}, can be expressed by the same LTL formula,
but taking $p$ to be the condition that $T$ is enabled, which holds for any state with an
outgoing transition from~$T$, and
$q$ to be the occurrence of a task $T$, a condition that holds when performing any
transition from~$T$. The whole notion of weak fairness (when given a collection~$\T$ of
tasks) is then the conjunction, for all $T\in\T$, of the formulas
$\textbf{G}(\textbf{G}(\textit{enabled}(T)) \Rightarrow \textbf{F}(\textit{occurs}(T)))$.

Likewise, \emph{Response to Persistence} is expressed as 
$\textbf{G}(\textbf{G}\textbf{F}p \Rightarrow \textbf{F}q)$, where
$\textbf{G}\textbf{F}p \Rightarrow \textbf{F}q$ is logically equivalent to
$\neg \textbf{G}(\textbf{F}p \wedge \neg q)$ \cite{GPSS80}.
\emph{Response to Persistence} can equivalently be expressed as
$\textbf{G}\textbf{F}p \Rightarrow \textbf{G}\textbf{F}q$, saying that
if condition $p$ holds infinitely often then response $q$ occurs infinitely often.
Again, strong fairness is expressed in the same way.

\section{Full Fairness}\label{sec:full}

In \citep[Def.~3]{QS83} a path $\pi= s_0\,t_1\,s_1\,t_2\,s_2\dots$ is regarded as unfair if there
exists a predicate $\He$ on the set of states such that a state in $\He$ is reachable from each state
$s_i$, yet $s_i\mathbin\in \He$ for only finitely many $i$.\linebreak[3] It is not precisely defined which sets of states
$\He$ can be seen as predicates; if we allow any set $\He$ then the resulting notion of fairness is not
feasible. 

\begin{exam}\label{ex:FFQ}
  Let $\pi$ be any infinite rooted path of the following program.
\[
\begin{array}{@{}c@{}}
\init{(x,y)}{(0,0)}\\[0.8ex]
\begin{array}{@{}l@{}}
{\bf while~}(\textit{true}){\bf ~do~}(x,y)\mathbin{:=}(x{+}1,0){\bf ~od}
\end{array}
~~
\|
~~
\begin{array}{@{}l@{}}
{\bf while~}(\textit{true}){\bf ~do~}(x,y)\mathbin{:=}(x{+}1,1){\bf ~od}
\end{array}
\end{array}
\]
  Then $\pi$ can be encoded as a function $f_\pi:\IN\rightarrow \{0,1\}$, expressing $y$ in
    terms of the current value of $x$.
  Let $\He:=\{(x,1{-}f_\pi(x)) \mid x \in\IN\}$.
  $\He$ is reachable from each state of $\pi$, but never reached.
  Consequently $\pi$ is unfair. This holds for all infinite rooted paths. So no path is fair.
\end{exam}

\noindent
Since we impose feasibility as a necessary requirement for a notion of fairness, this rules out the
form of fairness contemplated above. 
Nevertheless, this idea can be used to support liveness properties $\Ge$,
and in this form we propose to call it \emph{full fairness}.

\begin{definition}\label{df:AGEF}\cite{GV06}
A liveness property $\Ge$, modelled as a set of states in a transition system $G$, is an \emph{AGEF property}
iff (a state of) $\Ge$ is reachable from every state $s$ that is reachable from an initial state of $G$.\footnote{%
The name stems from the corresponding CTL-formula.}
\end{definition}

\begin{definition}\label{df:Fu}
  A liveness property $\Ge$ holds under the assumption of \emph{full fairness} (\Fu)
  iff $\Ge$ is an AGEF property.
\end{definition}
We cannot think of any formal definition that approximates the essence of the notion of fairness
  from \citep[Def.~3]{QS83} more closely then \df{Fu} above, while still making it feasible.

Full fairness is not a fairness assumption in the sense of \Sec{fairness}.
It does not eliminate a particular set of paths (the unfair ones) from consideration.
Nevertheless, like fairness assumptions, it increases the set of liveness properties that hold, and
as such can be compared with fairness assumptions.

Obviously no feasible fairness property can be strong enough to ensure a liveness property
that is not AGEF, for it is always possible to follow a finite path to a state from which it is
hopeless to still satisfy $\Ge$.
Hence full
fairness, as a tool for validating liveness properties, is the strongest notion of fairness conceivable.

\section{Strong Fairness of Transitions}\label{sec:SFTransitions}

This section provides an exact characterisation of the class of liveness properties that hold under
strong fairness of transitions (ST). It follows that on finite-state transition systems, 
strong fairness of transitions is as
 strong as full fairness.

\begin{theorem}\rm\label{thm:ST-fairness}
A liveness property $\Ge$, modelled as a set of states, holds in a transition system under the
assumption of strong fairness of transitions, iff $\Ge$ is an AGEF property and each infinite rooted path
that does not visit (a state of) $\Ge$ has a loop.
\end{theorem}

Note that the condition ``each infinite rooted path
that does not visit $\Ge$ has a loop'' is equivalent to
``each infinite rooted path that does not visit $\Ge$ contains a state that is visited infinitely often''.
For any infinite path $\pi$ that contains no state that is visited infinitely often can be transformed
into a path that has only states and transition from $\pi$, but skips all loops.

\begin{proof}
`$\Rightarrow$': Let $\Ge$ hold when assuming ST-fairness. 
  We first show, by contraposition,  that $\Ge$ must be an AGEF property. 
As pointed out above, if $\Ge$ is not an AGEF-property, it does not hold under any feasible fairness assumption, including ST.

Next we show that $\Ge$ satisfies the condition on infinite paths, again by contraposition.
Let $\pi$ be a loop-free infinite rooted path that does not visit (a state of) $\Ge$.
Since each state of $\pi$ is visited only finitely often, there is no transition that is enabled
relentlessly on $\pi$. Thus, $\pi$ is ST-fair, and consequently $\Ge$ does not hold under strong
fairness of transitions.

`$\Leftarrow$':
Let $\Ge$ be an AGEF-property and each infinite rooted path that does not visit
$\Ge$ has a loop. Let $\pi$ be an ST-fair rooted path.  It remains to show that $\pi$ visits $\Ge$.

  Define the \emph{distance} (from $\Ge$) of a state $s$ to be the length of the shortest path from
  $s$ to a state of $\Ge$. Since $\Ge$ is AGEF, any state $s$ of $\pi$ has a finite distance.
  For each $n\geq 0$ let $S_n$ be the set of states with distance $n$.
  Each state in $S_n$, with $n>0$, must have an outgoing transition to a state in $S_{n-1}$.
  \begin{itemize}
  \item Suppose that $\pi$ is finite. If its last state is in $S_n$ for some
  $n>0$, then a transition (to a state in $S_{n-1}$) must be
  enabled in its last state, contradicting the ST-fairness of $\pi$.
  Hence its last state is in $S_0$. So $\pi$ visits $\Ge$.
  \item Now suppose $\pi$ is infinite. Then, there is a state of $\pi$ that is visited infinitely often.
  Let $n\geq 0$ be such that a state $s$ in $S_n$ is visited infinitely often.
  Assume that $n>0$.
  Consider a transition $t$ from $s$ to a state in $S_{n-1}$. This transition is relentlessly
  enabled on $\pi$, and thus must occur in $\pi$ infinitely often. It follows that a state in
  $S_{n-1}$ is visited infinitely often as well. Hence, a trivial induction yields that a state in 
  $S_0$ is visited infinitely often. So $\pi$ visits $\Ge$.
\popQED
  \end{itemize}
\end{proof}

\begin{corollary}\rm\label{cor:ST-fairness}
A liveness property $\Ge$ holds in a finite-state transition system under the
assumption strong fairness of transitions iff $\Ge$ is an AGEF property.
\end{corollary}
Thus, on finite-state systems strong fairness of transitions is the strongest notion of fairness conceivable.
That it is strictly stronger than the other notions of Section~\ref{sec:taxonomy} is illustrated below.

\begin{exam}\label{ex:PF}
 Consider the process $X | c.X$, where $X \defis a.b.c.X$.  Let $\Ge$ consist of those states where
 both processes are in the same local state.  If we consider the finite-state transition system that
 lists all $9$ possible configurations of the system then $\Ge$ is an AGEF property, and thus will
 be reached under ST-fairness. Clearly, $\Ge$ is not guaranteed under $\St\Sy$- and
 $\St\Ac$-fairness.
\end{exam}

Example~\ref{ex:ST} shows that on infinite-state systems strong fairness of transitions is strictly
weaker than full fairness.

\section{Probabilistic Fairness}\label{sec:probabilistic}
When regarding nondeterministic 
choices as probabilistic choices with
unknown probability distributions, one could say that a liveness property
\emph{holds under the assumption of probabilistic fairness}
if it holds with probability 1, regardless of these probability distributions.
Formally, a \emph{probabilistic interpretation} of a transition system $G$\footnote{In this section
  we restrict attention to transition systems with exactly one initial state. An arbitrary transition
  system can be transformed into one with this property by adding a fresh initial state with an
  outgoing transition to each of the old initial states. This restriction saves us the effort of
  also declaring a probability distribution over the initial states.} is a function
$p:\Tr\rightarrow(0,1]$ assigning a positive probability to each transition, such that for each
state $s$ that has an outgoing transition the probabilities of all transitions with source $s$ sum up to 1.
Given a probabilistic interpretation, the probability of (following) a finite path is the product of
the probabilities of the transitions in that path. The probability of reaching a set of states $\Ge$
is the sum over all finite paths from the initial state to a state in $\Ge$, not passing through any states of
$\Ge$ on the way, of the probability of that path. Now $\Ge$ holds under probabilistic fairness ($\Pr$) iff
this probability is 1 for \emph{each} probabilistic interpretation.
Probabilistic fairness is a meaningful concept only for countably branching transition
systems, since it is impossible to assign each transition a positive probability such that their sum is 1
in case a state has uncountably many outgoing transitions.

A form of probabilistic fairness was contemplated in \cite{Pn83}.
Similar to full fairness, probabilistic fairness is not a fairness assumption in the sense of
\Sec{fairness}, as it does not eliminate a particular set of paths.
However, similar to fairness assumptions, it is a method for proving liveness properties.
This allows us to compare its strength with proper notions of fairness: one method for proving
living properties is called \emph{stronger} than another iff it proves a larger set of liveness properties. 
We show that probabilistic fairness is equally strong as strong fairness of transitions.

\begin{theorem}\rm\label{thm:probfair}
On countably branching transition systems a liveness property holds under probabilistic fairness iff
it holds under strong fairness of transitions.
\end{theorem}
\begin{proof}
We use the characterisation of strong fairness of transitions from Theorem~\ref{thm:ST-fairness}.
If a liveness property $\Ge$ is not AGEF, it surely does not hold under probabilistic fairness, since
there exists a state, reachable with probability $>0$, from which $\Ge$ cannot be reached.
Furthermore, if there exists an infinite loop-free rooted path that does not visit $\Ge$,
then the probabilities of the transitions leaving that path can be chosen in such a way that the
probability of remaining on this path forever, and thus not reaching $\Ge$, is positive.

Now let $\Ge$ be an AGEF-property that does not hold under probabilistic fairness.
It remains to be shown that there exists an infinite loop-free rooted path that does not visit $\Ge$.

Allocate probabilities to the transitions in such a way that the probability of reaching $\Ge$ from
the initial state is less than 1. For each state $s$, let $p_s$ be the probability of reaching $\Ge$
from that state. Since $p_s$ is the weighted average of the values $p_{s'}$ for all states $s'$
reachable from $s$ in one step, we have that if $s$ has a successor state $s'$ with $p_{s'}>p_s$,
then it also has a successor state $s''$ with $p_{s''}<p_s$.
\vspace{1ex}

\textsc{Claim:}
If, for a reachable state $s$, $p_s<1$ then there is a path $s t_1 s_1 t_2 s_2 \dots
s_n t_{n+1} s_{n+1}$ with $n\geq 0$ such that $p_{s_i}=p_s$ for $i=1,\dots,n$ and $p_{s_{n+1}}<p_s$.
\vspace{1ex}

\textsc{Proof of claim:}
Since $\Ge$ is an AGEF property, there is a path $\pi = s t_1 s_1 t_2 s_2 \dots s_k$
from $s$ to a state $s_k\in\Ge$. Clearly $p_{s_k}=1$. Let $n\in\IN$ be the smallest index for
which $p_{s_{n+1}} \neq p_{s_n}$. In case $p_{s_{n+1}} < p_{s_n}$ we are done. Otherwise, there must
be a successor state $s''$ of $s_n$ with $p_{s''}<p_{s_n}$, and we are also done.
\hfill \rule{1ex}{7pt}\vspace{1ex}

\textsc{Application of the claim:}
By assumption, $p_{s_0}<1$, for $s_0$ the initial state.
Using the claim, there exists an infinite path $\pi =s_0 t_1 s_1 t_2 s_2 \dots$ such that $p_{s_i}\leq p_s$
for all $i\in\IN$, and $\forall i\mathbin\in\IN.~ \exists j\mathbin>i.~p_{s_j}<p_{s_i}$. Clearly, this path does not
visit a state of $\Ge$, and no state $s$ can occur infinitely often in $\pi$. After cutting out
loops, $\pi$ is still infinite, and moreover loop-free, which finishes the proof.
\end{proof}

\section{Extreme Fairness}\label{sec:extreme}

In \cite{Pn83}, {\sc Pnueli} proposed the strongest possible
notion of fairness that fits in the
format of \df{fair}. 
A first idea (not {\sc Pnueli}'s) to define the strongest notion might be to admit \emph{any} set of transitions as a task.
However, the resulting notion of fairness would fail the criterion of feasibility.

\begin{exam}\label{ex:EF}
  The transition system for the program of Example~\ref{ex:FFQ}
  can be depicted as an infinite binary tree.
  For any path $\pi$ in this tree, let $T_\pi$ be the set of all transitions that do not occur in $\pi$. 
  On $\pi$, this task is perpetually enabled, yet does not occur.
  It follows that $\pi$ is unfair. As this holds for any path $\pi$, no path is fair, and the
  resulting notion of fairness is infeasible.
\end{exam}
Avoiding this type of example, {\sc Pnueli} defined the notion of \emph{extreme fairness} by admitting any
task (\ie\ any predicate on the set of transitions) that is definable in first-order logic.
This makes sense when we presuppose any formal language for defining predicates on transitions.
{\sc Pnueli} is not particularly explicit about the syntax and semantics of such a language, which is
justified because one merely needs to know that such formalisms can generate
only countably many formulas, and thus only countably many tasks.
With \thm{feasibility} it follows that extreme fairness is feasible.

Assuming a sufficiently expressive language for specifying tasks, any task according to fairness of
actions, transitions, instructions, etc., is also a task according to extreme fairness.
Hence the strong version of extreme fairness ($\Ex$)---the one considered in \cite{Pn83}---sits at the
top of the hierarchy of \fig{taxonomy}, but still below full fairness.

\begin{exam}\label{ex:FF}
  Under full fairness the following program will surely terminate.
\[
\begin{array}{@{}c@{}}
\init{y}{1}\\[0.8ex]
\left.
\begin{array}{@{}l@{}}
{\bf while~}(y>0){\bf ~do~~~}y:=y-1{\bf ~~~od}
\end{array}
~~~
\right \|
~~~
\begin{array}{@{}l@{}}
{\bf while~}(y>0){\bf ~do~~~}y:=y+1{\bf ~~~od}
\end{array}
\end{array}
\]
   Yet, under any other notion of fairness, including probabilistic fairness, it is possible that $y$
   slowly but surely increases (in the sense that no finite value is repeated infinitely often).
   Fairness (even extreme fairness) can
   cause a decrease of  $y$ once in a while, but is not strong enough to counter the upwards trend.
\end{exam}
{\sc Pnueli} invented extreme fairness as a true fairness notion
that closely approximates a variant of probabilistic fairness \cite{Pn83}.
His claim that for finite-state programs extreme fairness and probabilistic fairness coincides, now
follows from Corollary~\ref{cor:ST-fairness} and Theorem~\ref{thm:probfair}.

\section{Justness}\label{sec:justness}

Fairness assumptions can be a crucial ingredient in the verification of liveness properties
of real systems. A typical example is the verification of a communication
protocol that ensures that a stream of messages is relayed correctly, without loss or
reordering, from a sender to a receiver, while using an unreliable
communication channel.  The \emph{alternating bit protocol} \cite{Lyn68,BSW69}, an example of such
a communication protocol, works by means of acknowledgements,
and resending of messages for which no acknowledgement is received.

To prove correctness of such protocols one has to make the indispensable fairness assumption 
that attempts to transmit a message over the unreliable channel will not fail in perpetuity.
Without a fairness assumption,
no such protocol can be found correct, and one misses a chance to
check the protocol logic against possible flaws that have nothing to
do with a perpetual failure of the communication channel.
For this reason, fairness assumptions are made in many process-algebraic
verification platforms, and are deeply ingrained in their methodology \cite{BK86,BBK87a}.
The same can be said for other techniques, such as automata-based approaches, or temporal logic.

Fairness assumptions, however, need to be chosen with care,
since they can lead to false conclusions.

\begin{exam}\label{ex:phone}
Consider\hfill the\hfill  process\hfill  $P:=(X|b)\backslash b$\hfill  where\hfill  $X \defis a.X + \bar b.X$.\hfill  The\hfill  process\hfill  $X$\hfill  models 
{\makeatletter
\let\par\@@par
\par\parshape0
\everypar{}\begin{wrapfigure}[3]{r}{0.267\textwidth}
 \vspace{-1.6ex}
  \input{phone}
  \centerline{\raisebox{1ex}{\box\graph}}
   \vspace{-10ex}
  \end{wrapfigure}
\noindent the behaviour of relentless Alice, who either picks up her phone 
when Bob is calling ($\bar b$), or performs another activity ($a$), such as eating an apple.
In parallel, $b$ models Bob, constantly trying to call Alice; the action $b$ models the 
call that takes place when Alice answers the phone. Only
a successful call can end Bob's misery. Our liveness property $\Ge$ is to achieve a connection between Alice and Bob.
\par}
  Under each of the notions of fairness of \Sects~\ref{sec:taxonomy}--\ref{sec:extreme}, $P$ satisfies $\Ge$.
  Yet, it is perfectly reasonable that the connection is never
  established: Alice could never pick up her phone, as she is not in the mood of talking to Bob; maybe she
  totally broke up with him.
\end{exam}
The above example is not an anomaly; it describes the default situation.
A fairness assumption says, in essence, that if one tries something often enough, one will eventually
succeed. There is nothing in our understanding of the physical universe that supports such a belief.
Fairness assumptions are justified in exceptional situations, the verification of the alternating
bit protocol being a good example. However, by default they are unwarranted.

In the remainder of this section we investigate what is needed to verify realistic liveness properties when not making any fairness assumptions.
\begin{exam}\label{ex:just}
  Consider the process $P:=X|c$ where $X \defis b.X$. Here 
{\makeatletter
\let\par\@@par
\par\parshape0
\everypar{}\begin{wrapfigure}[2]{r}{0.267\textwidth}
 \vspace{-5.5ex}
  \input{just}
  \centerline{\raisebox{1ex}{\box\graph}}
   \vspace{-10ex}
  \end{wrapfigure}
\noindent $X$ is an (endless) series of calls between Alice and Bob in London (who finally speak to each other again), 
while $c$ is the action by Cateline of eating a croissant in Paris.
Clearly, there is no interaction between Cateline's desire and Alice's and Bob's chats.
The question arises whether she is guaranteed to succeed in eating her breakfast.
Using progress, but not fairness, the answer is negative,
for there is a progressing rooted path consisting of $b$-transitions only.
Nevertheless, as nothing stops Cateline from making progress, in reality $c$ will occur.
\par}
\end{exam}

\noindent
We therefore propose a strong progress assumption, called \emph{justness}:
\vspace{-2ex}

\begin{equation}\label{J}\tag{\,J}
\parbox{.9\textwidth}{\it
  Once a transition is enabled that stems from a set of parallel components,
  one (or more) of these components will eventually partake in a transition.}
  \end{equation}
In \ex{just}, justness would guarantee that Cateline will have breakfast.

To formalise justness, we employ a binary relation $\naconc$ between transitions,
with $t \naconc u$ meaning that the transition $u$ interferes with $t$, in the sense that
it affects a resource that is needed by $t$. 
In particular, $t$ cannot run concurrently with $u$.

\begin{definition}\label{df:justness}
  A path $\pi$ in a transition system representing a closed system is \emph{just} if for each
  transition $t$ with $s\mathop{:=}\source(t)\mathop{\in}\pi$, a transition $u$ occurs
  in $\pi$ past the occurrence of $s$,
  such that $t \naconc u$.
\end{definition}

\noindent
When thinking of the resources alluded to in the above explanation of $\naconc$ as components,
$t \naconc u$ means  that a component affected by the execution of $u$ is a necessary
participant in the execution of $t$.
The relation $\naconc$ can then be expressed in terms of two functions 
$\ac,\pc:\Tr\rightarrow\Pow(\Ce)$ telling for each transition $t$ which components
  are necessary participants in the execution of $t$, and
  which are affected by the execution of $t$, respectively.
  Then
\[ t \aconc u \quad\mbox{iff}\quad \ac(t)\cap\pc(u)=\emptyset\,.\]
We assume the following crucial property for our transition systems:
\begin{equation}\label{4}\tag{\#}
\parbox{.9\textwidth}{If $t,u\mathbin\in\Tr$ with $\source(t)=\source(u)$ and $\ac(t)\cap\pc(u)=\emptyset$ then
  there is a transition $v\in \Tr$ with $\source(v)=\target(u)$ and $\ac(v)=\ac(t)$.}
  \end{equation}
The underlying idea is that if a transition $u$ occurs that has no influence on components in $\ac(t)$,
then the internal state of those components is unchanged, so any synchronisation between these
components that was possible before $u$ occurred, is still possible afterwards.

We also assume that $\ac(t)\cap\pc(t)\neq\emptyset$ for all $t\in\Tr$, that is, $\naconc$ is reflexive.
Intuitively, this means that a least one component that is required by $t$ is also affected when taking the transition $t$.
Thus, one way to satisfy the justness requirement is by executing $t$.

Justness is thus precisely formalised for any transition system arising from a process algebra,
Petri net, or other specification formalism for which functions $\pc$ and $\ac$, satisfying \eqref{4} and
reflexivity of $\naconc$,
are defined. 
Many specification formalisms satisfy $\pc(t)=\ac(t)$ for all $t\mathbin\in\Tr$.
In such cases we write $\comp$ for $\pc$ and $\ac$, and $\conc$ for $\aconc$, this relation then being
symmetric.

This applies in particular to the fragment of CCS studied in this paper; 
in \App{fragment} we present a definition of $\comp$ which satisfies  \eqref{4} and reflexivity of $\not\conc$.
For a `local'
transition
$\comp$ returns a singleton set, 
whereas for a CCS-synchronisation yielding a $\tau$-transition two components will be present.
We only include minimal components (cf.\ Footnote~\ref{minimal}).
The transition \plat{$(P|a.Q)\backslash c\, |\bar a.R\linebreak[1] \goto{\tau} (P|Q)\backslash c\, |R$}
for example stems from a synchronisation of $a$ and $\bar a$ in the components $a.Q$ and $\bar a.R$.
Since $a.Q$ is a subcomponent of $(P|a.Q)\backslash c$, it can be argued that the component
$(P|a.Q)\backslash c$ also partakes in this transition $t$. Yet, we do not include this component in
$\comp(t)$.

The justness  assumption is fundamentally different from a fairness assumption: rather than
assuming that some condition holds perpetually, or infinitely often, we merely assume it to hold once.
In this regard, justness is similar to progress.
Furthermore, justness \eqref{J} implies progress: we have ${\rm P} \preceq \J$.

Assuming justness ensures $\Ge$ in \exs~\ref{ex:progress},~\ref{ex:SA},~\ref{ex:ST} and \ref{ex:just},
but not in \exs~\ref{ex:SG}, \ref{ex:SC}, \ref{ex:WS}, \ref{ex:WC} and \ref{ex:phone}.

In Examples~\ref{ex:weak} and~\ref{ex:mutex} 
there are three components: \Left(eft), \R(ight) and an implicit memory (\M),
where the value of $y$ is stored.
The analysis of \ex{weak} depends on the underlying memory model.
Clearly, any write to $y$ affects the memory component, \ie\ $\M\in\pc(y:=0)$
and $\M\in\pc(r)$, where $r$ is the atomic instruction ``${\bf while~}(y>0){\bf ~do~~~}y:=y+1{\bf ~~~od}$''.
Moreover, since $r$ involves reading from memory, {\M} is a necessary participant in this transition:
$\M\in\pc(r)$.
If we would assume that a write always happens immediately, regardless of the state of the memory,
one could argue that the memory is not a necessary participant in the transition $y:=0$.
This would make justness a sufficient assumption to ensure $\Ge$.
A more plausible memory model, however, might be that no two components can successfully write to the
same memory at the same time. For one component to write, it needs to `get hold of' the memory,
and while writing, no other component can get hold of it. Under this model
$\M\in\ac(y:=0)$, so that justness is not a sufficient assumption to ensure $\Ge$.

In \ex{mutex}, by the above reasoning, we have $\pc(\ell_1)=\ac(\ell_1)=\pc(m_1)=\ac(m_1)=\{\Left,\M\}$.
 Since $\ac(\ell_1)\cap\pc(m_1)\neq\emptyset$,
we have $\ell_1 \naconc m_1$, \ie\
the occurrence of $m_1$ can disable $\ell_1$.
Justness does not guarantee that $\ell_1$ will ever occur, so assuming justness is insufficient to
ensure $\Ge$.

The difference between $\ac$ and $\pc$ shows up in process algebras  featuring broadcast
communication, such as the one of \cite{GH15a}.
If $t$ models a broadcast transition,
the only necessary participant is the component that performs
the transmitting part of this  synchronisation.
However, all components that receive the broadcast are in~$\pc(t)$.
In this setting, one may expect that necessary participants in a synchronisation are always affected
($\ac(t)\subseteq\pc(t)$).
However, in a process algebra with signals \cite{Bergstra88,EPTCS255.2} we find transitions $t$ with
$\ac(t)\not\subseteq\pc(t)$. Let $t$ model the action of a driver seeing a red traffic light.
This transition has two necessary participants: the driver and the traffic light.
However, only the driver is affected by this transition; the traffic light is not.

In applications where $\ac\neq\pc$, two variants of fairness of components and fairness of groups
of components can be distinguished, namely by taking either $\ac$ or $\pc$ to be the function
$\comp:\Tr\rightarrow\Pow(\Ce)$ assumed in \Sec{taxonomy}. Since the specification formalisms
dealing with fairness found in the literature, with the exception of our own work \cite{TR13,GH15a,EPTCS255.2},
did not give rise to a distinction between necessary and affected
participants in a synchronisation, in our treatment of fairness we
assume that $\ac = \pc \ (= \comp)$,
leaving a treatment of the general case for future work.

Special cases of the justness assumption abound in the literature.
{\sc Kuiper \& de Roever} \cite{KdR83} for instance assume `fundamental liveness',
attributed to {\sc Owicki \& Lamport} \cite{OL82},
``ensuring that if a process is continuously able to proceed, it eventually will.''
Here `process' is what we call `component'.
This is formalised for specific commands of their process specification language, and in each case
appears to be an instance of what we call justness. However, there is no formalisation of
justness as one coherent concept. Likewise, the various liveness axioms in \cite{OL82}, starting
from Section~5, can all be seen as special cases of justness.
{\sc Apt, Francez \& Katz} \cite{AFK88} assume `minimal progress', also attributed to \cite{OL82}:
``Every process in a state with enabled \emph{local} actions will eventually execute some action.''
This is the special case of \df{justness} where $\ac(t)$ is a singleton.

In the setting of Petri nets the (necessary and affected) participants in a transition $t$ could be
taken to be the preplaces of $t$. Hence $t\nconc u$
means that the transition system transitions $t$ and $u$ stem from Petri net
transitions that have a common preplace. Here justness says that if a Petri net transition $t$ is
enabled, eventually a transition will fire that takes away a token from one of the preplaces of $t$.
In Petri-net theory this is in fact a common assumption \cite{Rei13}, which is often made
without giving it a specific name.

In \cite{TR13,GH15a} we introduced justness in the setting
of two specific process algebras. Both have features that resist a correct
translation into Petri nets, and hence our justness assumptions cannot be explained in terms of
traditional assumptions on Petri~nets.
The treatment above is more general, not tied to any specific process algebra.

In \cite{GH15b,EPTCS255.2} we argue that a justness assumption is essential for the validity of any
mutual exclusion protocol. It appears that in the correctness arguments of famous mutual exclusion
protocols appearing in the literature \cite{Pe81,bakery} the relevant justness assumption is taken
for granted, without explicit acknowledgement that the correctness of the protocol rests on the
validity of this assumption. We expect that such use of justness is widespread.

\section{Uninterrupted Justice}\label{sec:J-fairness}

In \ex{phone} weak fairness (of any kind) is enough to ensure $\Ge$.\vspace{1pt}
Now, let the action $a$ of calling anybody except Bob consist of two subactions
$a_1.a_2$, yielding the specification \plat{$X'\defis a_1.a_2.X' + b.X'$}.
Here $a_1$ could for instance be
the act of looking up a telephone number in the list of contacts, 
and $a_{2}$ the actual calling attempt.
Substituting $a_1.a_2$ for $a$ is a case of \emph{action refinement}~\cite{GG01}, and represents a
change in the level of abstraction at which actions are modelled.
After this refinement, the system no longer satisfies $\Ge$ when assuming weak fairness (only).
Hence weak fairness is not preserved under refinement of actions.

When the above is seen as an argument against using weak fairness, two alternatives come to mind.
One is to use strong fairness instead; the other is to tighten the definition of weak fairness in
such a way that that a rooted path is deemed fair only if it remains weakly fair after any action
refinement.

To this end we revisit the meaning of ``perpetually enabled'' in \df{fair}.
This way of defining weak fairness stems from {\sc Lehmann, Pnueli \& Stavi}~\cite{LPS81},
who used the words ``justice'' for weak fairness and ``continuously'' for perpetual.
Instead of merely requiring that a task is enabled in every state,
towards a more literal interpretation of ``continuously'' we additionally require that it
remains enabled between each pair of successive states, e.g.\ during the execution of the transitions.

Formalising this requires us to define enabledness during a transition.
If $\source(t)=\source(u)$ then $t \aconc u$ tells us that
the possible occurrence of $t$ is unaffected by the possible occurrence of $u$.
In other words, $t \naconc u$ indicates that the occurrence of $u$ ends or interrupts the enabledness
of $t$, whereas $t\aconc u$ indicates that $t$ remains enabled during the execution of $u$.

\begin{definition}\label{df:justice}
For a transition system $G\mathbin=(S,\Tr,\source,\target,I,\T)$, equipped with a relation
${\aconc} \subseteq \Tr\times\Tr$, a task $T\in\T$ is
\emph{enabled} during a transition $u\in \Tr$ if $t \aconc u$ for some $t \in T$ with
$\source(t)=\source(u)$.
It is said to be \emph{continuously enabled} on a path $\pi$ in $G$, if it is enabled in every state
\emph{and transition} of $\pi$.
A path $\pi$ in $G$ is \emph{\J-fair} if, for every suffix $\pi'$ of $\pi$,
each task that is continuously enabled on $\pi'$, occurs in $\pi'$.
\end{definition}

\noindent
As before, we can reformulate the property to avoid the quantification over suffixes:
A path $\pi$ in $G$ is \J-fair if each task that from some state onwards is continuously enabled on $\pi$, occurs infinitely often in $\pi$.

Letting $\X$ range over $\{\J,\W,\St\}$ extends the taxonomy of \Sec{taxonomy}
by six new entries, some of which coincide (\J\Sy, {\J\Gr} and \J). 
Figure~\ref{full taxonomy} shows the full hierarchy, including 
the entries SWI from \Sec{swi}, \Ex, $\Pr$ and $\Fu$ from \Sects~\ref{sec:full}--\ref{sec:extreme}
and J from \Sec{justness}.%
\begin{figure}[th]
\input{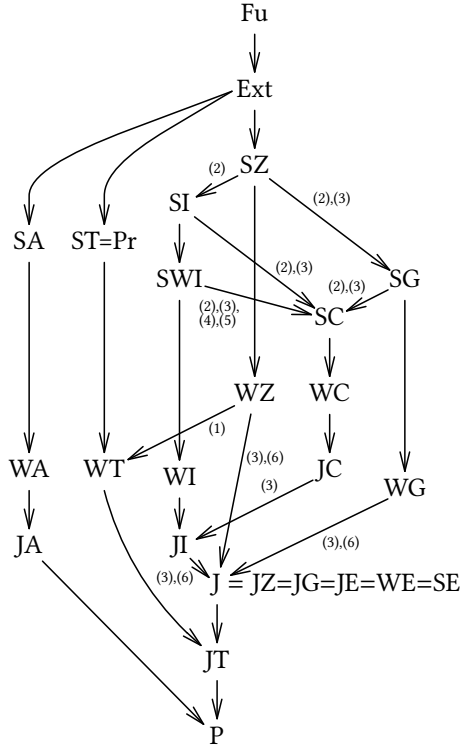}
\centerline{\raisebox{1ex}{\box\graph}}
\caption{A classification of progress, justness and fairness assumptions}
\label{full taxonomy}
\end{figure}
The  arrows are valid under the assumptions \ref{unique synchronisation}--\ref{cmp}
from \Sec{taxonomy}, \ref{swi1}--\ref{swi2} from \Sec{adding}, and the assumption that $\ac = \pc = \comp$, together with
\begin{enumerate}[(1)]
  \vspace{1ex}
  \setcounter{enumi}{5}
  \item
  if $t\mathbin{\conc} u$ with $\source(t)\mathbin=\source(u)$, then $\exists v\mathbin\in\Tr$
  with $\source(v)\mathbin=\target(u)$ and $\instr(v)\mathbin=\instr(t)$.\label{conc}%
  \vspace{1ex}
\end{enumerate}
This assumption says that if a transition $t$ is enabled that does not depend on any component that is affected by $u$, and $u$
occurs, then afterwards a variant of $t$, stemming from the same instructions, is still possible.
This assumption strengthens \eqref{4} and  holds for the transition systems derived from the
fragment of CCS presented in \App{fragment}.
 
Clearly, ${\rm P} \preceq \J\Y \preceq \W\Y \preceq \St\Y$ for $\Y \in \{\Ac,\Ts,\In,\Sy,\Cp,\Gr\}$.

\begin{proposition}\rm
When assuming \ref{cmp} then $\J\In \preceq \J\Cp$.
\end{proposition}

\begin{proof}
We show that any $\J\Cp$-fair path is also $\J\In$-fair.
W.l.o.g.\ let $\pi$ be a $\J\Cp$-fair path on which a task $T_I$, for a given instruction $I\mathbin\in\I$, is
continuously enabled. Then, using assumption \ref{cmp}, on $\pi$ component $\cmp(I)$ is continuously
enabled. Hence, the task $T_{\cmp(I)}$ will occur in $\pi$. Let $t\mathop\in T_{\cmp(I)}$ be a transition in $\pi$.
Then $\cmp(I) \mathbin\in \comp(t) \cap \comp(u)$ for any $u\mathbin\in T_I$, so that \mbox{$t \nconc u$}.
Thus, no transition from $T_I$ is enabled during the execution of $t$, contradicting the
assumption that $T_I$ is continuously enabled on $\pi$.
Hence $\pi$ is $\J\In$-fair.
\end{proof}

\begin{proposition}\rm
When assuming \ref{cmp} and \ref{conc} then
$\J \preceq \J\Y$ for $\Y \in \{\In,\Sy,\Cp,\Gr\}$.
\end{proposition}
\begin{proof}
We first show that any $\J\In$-fair path is also just.
  W.l.o.g.\ let $\pi$ be a $\J\In$-fair path such that a transition $t$ is enabled in its first state and let $I\in\instr(t)$.
  To construct a contradiction, assume that $\pi$ contains no transition $u$ with $\comp(t)\cap\comp(u)\neq\emptyset$.
  Then, using \ref{conc}, for each transition $u$ occurring in $\pi$, a transition from $T_I$ is
  enabled during and right after $u$. Hence, by $\J\In$-fairness, the task $T_I$ must occur in
  $\pi$. So $\pi$ contains a transition $u$ with $I\in\instr(u)$, and thus
    $\cmp(I)\in\comp(t)\cap\comp(u)$, using \ref{cmp}, 
  but this contradicts our assumption.
It follows that $\pi$ is just.
  
  The proofs for $\Y \mathbin\in \{\Sy,\Cp,\Gr\}$ are similar.
\end{proof}

\begin{proposition}\rm
When assuming  \ref{cmp} then $\J\Y \preceq \J$ for $\Y \in \{\Ts,\Sy,\Gr\}$.
\end{proposition}
\begin{proof}
We first show, by contradiction, that any just path is also $\J\Gr$-fair.
W.l.o.g.\ let $\pi$ be a just path on which a task $T_G$, for a given $G\subseteq\Ce$, is
continuously enabled, yet never taken.\footnote{We will arrive at a contradiction even without
  using that $T_G$ is never taken on $\pi$.}
Let $t\in T_G$ be enabled in the first state of $\pi$.
Then, by justness, $\pi$ contains a transition $u$ with $\emptyset\neq\comp(u)\cap\comp(t)=\comp(u)\cap G$.
It follows that no transition $v\in T_G$ is enabled during $u$, contradicting the assumption 
of continuously enabledness.

Next we prove, again by contradiction, that any just  path is also $\J\Sy$-fair.
For that we assume a task $T_Z$, for a given $Z\subseteq\I$, that is continuously enabled on $\pi$.
Using \ref{cmp}, let $G=\{\cmp(I) \mid I \in Z\}$. Then also $T_G$ is continuously enabled on $\pi$,
which led to a contradiction above.

Last we look at \J\Ts-fairness. We assume a transition $t$ is continuously enabled on a just path $\pi$.
Let $G=\comp(t)$. Then also $T_G$ is continuously enabled on $\pi$,
which led to a contradiction above.
\end{proof}

The absence of any further arrows follows from \exs~\ref{ex:mutex}--\ref{ex:phone}:
  \ex{phone} separates J-fairness from weak fairness; it
  shows that $\W\Y \not\preceq \J\Y'$ (and hence $\St\Y \not\preceq \J\Y'$)
  for all $\Y,\Y' \in \{\Ac,\Ts,\In,\Sy,\Cp,\Gr\}$.
  Example~\ref{ex:SA}, in which $\J\Ts$-fairness suffices to ensure $\Ge$,
  shows that there are no further arrows from SA, and thus none from WA and JA\@.
  Example~\ref{ex:ST}, in which justness or $\J\Ac$-fairness suffices to ensure $\Ge$,
  shows that there are no further arrows from ST, and thus none from WT, JT and P.
  Example~\ref{ex:WC}, in which $\J\Cp$-fairness suffices to ensure $\Ge$,
  shows that there are no further arrows into $\J\Cp$.
  So, using transitivity, and remembering the results from \Sects~\ref{sec:taxonomy}--\ref{sec:full},
  it suffices to show that $\J\Ac \not\preceq \St\Sy$,
  $\J\In \not\preceq \W\Sy$ and $\J\In \not\preceq \W\Gr$.
  This will be demonstrated by the following two examples.
 
\begin{exam}\label{ex:JA}
Consider the process $(X | Y | Z)\backslash b\backslash d$\vspace{1pt}
where  \plat{$X \mathop{\defis}$}
{\makeatletter
\let\par\@@par
\par\parshape0
\everypar{}\begin{wrapfigure}[8]{r}{0.318\textwidth}
 \vspace{-5.5ex}
  \input{JA}
  \centerline{\raisebox{1ex}{\box\graph}}
  \end{wrapfigure}
  \noindent \plat{$\bar b.\bar b.X[f] {+} a.X {+} c.X$},\plat{$Y \hspace{-.5pt}\defis\hspace{-.5pt} \bar d.\bar d.Y[f] {+} a.X {+} c.X$} and \plat{$Z \hspace{-.5pt}\defis\hspace{-.5pt} b.b.d.d.Z$},
using a relabelling operator with $f(a)=c$ and $f(c)=a$.
Under $x\Ac$-fairness, for $x\in\{\J,\W,\St\}$, the path $(a\tau)^\infty$ is unfair, because task $c$ is perpetually
enabled, yet never taken. However, there are only twelve instructions (and three components), each of
which gets its fair share. Also, each possible synchronisation will occur. 
So under $\X\Y$-fairness for $\Y \in\{\In,\Sy,\Cp,\Gr\}$ the path is fair.
\par}
\end{exam}

\begin{exam}\label{ex:JI}
Consider the process $(e | X | Y | Z)\backslash a\backslash b\backslash c\backslash d\backslash e$\vspace{1pt}
where \plat{$X \defis \bar a.\bar b.X + \bar e$}, \plat{$Y \defis \bar c.\bar d.Y + \bar e$} and \plat{$Z \defis a.b.c.d.Z$}.  Its
transition system is the same as for \ex{WC}.
Under $\X\In$- and under $\X\Cp$-fairness $\Ge$ is satisfied, but under $\W\Y$-fairness with $\Y\mathbin\in\{\Ac,\Ts,\Sy,\Gr\}$ it is not.
\end{exam}
 
\section{Fairness of Events}\label{sec:events}
 
An \emph{event} in a transition system can be defined as an equivalence class of transitions
stemming from the same synchronisation of instructions, except that each visit of an instruction
gives rise to a different event. Fairness of events can best be explained in terms of examples.

In \exs~\ref{ex:SA} and~\ref{ex:ST} $\Ge$ does hold: on the infinite path that violates $\Ge$
there is a single event corresponding with the action $a$, which is perpetually enabled.

In \ex{SG} $\Ge$ does not hold, for each round through the recursive equation counts as a separate
visit to the $b$-instruction,
so during the run $a^\infty$ no $b$-event is enabled
more than once. Similar reasoning applies to \exs~\ref{ex:SC}--\ref{ex:WC},
where $\Ge$ is not ensured and/or the indicated path is SE-fair.

Fairness of events is defined in {\sc Costa \& Stirling} \cite{CS84}---although on a
restriction-free subset of CCS where it coincides with fairness of components---and in
{\sc Corradini, Di Berardini \& Vogler}~\cite{CDV06a}, 
named fairness~of~actions.

Capturing fairness of events in terms of tasks, as in \Sec{taxonomy}, with the events as tasks, requires a semantics of CCS
different from the standard one given in \App{CCS}, yielding a partially unfolded transition system.
Example~\ref{ex:SG}, for instance, needs an infinite transition system to express that after each $a$-transition
a different $b$-event is enabled, and likewise for \exs~\ref{ex:SA}, \ref{ex:SC}, \ref{ex:WS} and \ref{ex:WC}.
Such a semantics, involving an extra layer of labelling, appears in \cite{CS84,CS87,CDV06a}.
What we need to know here about this semantics is that
\begin{enumerate}[(1)]
  \setcounter{enumi}{6}
  \item the events form a partition of the set of transitions;\label{partition}
  \item if an event $e$ is enabled in two states $s$ and $s'$ on a path $\pi$,
    then $e$ is enabled also in each state and during each transition that occurs between $s$ and $s'$;\label{event}
  \item if $t\conc u$ with $\source(t)\mathbin=\source(u)$ and $t \mathbin\in e$, then
    $\exists v\mathbin\in\Tr$ with $\source(v)\mathbin=\target(u)$ and $v\mathbin\in e$,\label{concevent}
  \item if $t,t'\in e$ then $\comp(t)=\comp(t')$. \label{compevent}
    So let $\comp(e):=\comp(t)$ when $t\in e$.
\end{enumerate}
Property \ref{event} implies immediately that strong fairness of events (SE), weak fairness of events (WE) and J-fairness of
events (JE) coincide.

\begin{theorem}\rm
A path is just iff it is JE-fair.
\end{theorem}

\begin{proof}
  `$\Rightarrow$':
  The proof is by contradiction.
  Let $\pi$ be a just path on which an event $e$ is continuously enabled, yet never taken.%
  \setcounter{footnote}{13}\footnotemark{}
  By justness, a transition $u$ occurs in $\pi$ such that $\comp(e)\cap\comp(u)\mathbin{\neq}\emptyset$.
  By assumption, event $e$ must be enabled during $u$. Hence there is a $t\in e$ with $t \conc u$, \ie\ $\comp(t)\cap\comp(u)=\emptyset$,
  in contradiction with \ref{compevent}.  Hence $\pi$ is JE-fair.
  
  `$\Leftarrow$':
  W.l.o.g.\ let $\pi$ be a JE-fair path such that a transition $t$ is enabled in its first state.
  Using \ref{partition}, let $e$ be the event with $t \in e$.
  To construct a contradiction, assume that $\pi$ contains no transition $u$ with $\comp(t)\cap\comp(u)\neq\emptyset$.
  Then, using \ref{concevent}, $e$ is continuously enabled on $\pi$.
  So, by JE-fairness, the task $e$ must occur in $\pi$, \ie $\pi$ contains a transition $v \in e$.
  By \ref{compevent}, $\comp(v)=\comp(t) \neq\emptyset$, but this contradicts the assumption. So $\pi$ is just.
\end{proof}

\section[Reactive Systems]{Reactive Systems}\label{sec:reactive}
Sections~\ref{sec:liveness}--\ref{sec:events} dealt with closed systems,
having no run-time interactions with the environment. We now generalise almost all definitions and
results to reactive systems, interacting with their environments through synchronisation of actions.

\begin{exam}
  Consider the CCS process $a.\tau$, represented by the transition system\\[1mm]  
{\expandafter\ifx\csname graph\endcsname\relax
   \csname newbox\expandafter\endcsname\csname graph\endcsname
\fi
\ifx\graphtemp\undefined
  \csname newdimen\endcsname\graphtemp
\fi
\expandafter\setbox\csname graph\endcsname
 =\vtop{\vskip 0pt\hbox{%
\pdfliteral{
q [] 0 d 1 J 1 j
0.576 w
0.576 w
28.8 -4.32 m
28.8 -6.70587 26.86587 -8.64 24.48 -8.64 c
22.09413 -8.64 20.16 -6.70587 20.16 -4.32 c
20.16 -1.93413 22.09413 0 24.48 0 c
26.86587 0 28.8 -1.93413 28.8 -4.32 c
S
Q
}%
    \graphtemp=.5ex
    \advance\graphtemp by 0.060in
    \rlap{\kern 0.340in\lower\graphtemp\hbox to 0pt{\hss {\scriptsize 1}\hss}}%
\pdfliteral{
q [] 0 d 1 J 1 j
0.576 w
0.072 w
q 0 g
12.96 -2.52 m
20.16 -4.32 l
12.96 -6.12 l
12.96 -2.52 l
B Q
0.576 w
0 -4.32 m
12.96 -4.32 l
S
100.8 -4.32 m
100.8 -6.70587 98.86587 -8.64 96.48 -8.64 c
94.09413 -8.64 92.16 -6.70587 92.16 -4.32 c
92.16 -1.93413 94.09413 0 96.48 0 c
98.86587 0 100.8 -1.93413 100.8 -4.32 c
S
Q
}%
    \graphtemp=.5ex
    \advance\graphtemp by 0.060in
    \rlap{\kern 1.340in\lower\graphtemp\hbox to 0pt{\hss {\scriptsize 2}\hss}}%
\pdfliteral{
q [] 0 d 1 J 1 j
0.576 w
0.072 w
q 0 g
84.96 -2.52 m
92.16 -4.32 l
84.96 -6.12 l
84.96 -2.52 l
B Q
0.576 w
28.8 -4.32 m
84.96 -4.32 l
S
Q
}%
    \graphtemp=\baselineskip
    \multiply\graphtemp by -1
    \divide\graphtemp by 2
    \advance\graphtemp by .5ex
    \advance\graphtemp by 0.060in
    \rlap{\kern 0.840in\lower\graphtemp\hbox to 0pt{\hss $a$\hss}}%
\pdfliteral{
q [] 0 d 1 J 1 j
0.576 w
172.8 -4.32 m
172.8 -6.70587 170.86587 -8.64 168.48 -8.64 c
166.09413 -8.64 164.16 -6.70587 164.16 -4.32 c
164.16 -1.93413 166.09413 0 168.48 0 c
170.86587 0 172.8 -1.93413 172.8 -4.32 c
h q 0.7 g
B Q
Q
}%
    \graphtemp=.5ex
    \advance\graphtemp by 0.060in
    \rlap{\kern 2.340in\lower\graphtemp\hbox to 0pt{\hss {\scriptsize 3}\hss}}%
\pdfliteral{
q [] 0 d 1 J 1 j
0.576 w
0.072 w
q 0 g
156.96 -2.52 m
164.16 -4.32 l
156.96 -6.12 l
156.96 -2.52 l
B Q
0.576 w
100.8 -4.32 m
156.96 -4.32 l
S
Q
}%
    \graphtemp=\baselineskip
    \multiply\graphtemp by -1
    \divide\graphtemp by 2
    \advance\graphtemp by .5ex
    \advance\graphtemp by 0.060in
    \rlap{\kern 1.840in\lower\graphtemp\hbox to 0pt{\hss $\tau$\hss}}%
    \hbox{\vrule depth0.120in width0pt height 0pt}%
    \kern 2.400in
  }%
}%

\centerline{\raisebox{2ex}{\box\graph}~.}}\vspace{1ex}
\noindent
Here $a$ is the action of receiving the signal $\bar a$ from the environment.
Will this process satisfy the liveness property $\Ge$, by reaching state 3?
When assuming the progress property for closed systems proposed in \Sec{progress}, the answer is positive.
However, when taking the behaviour of the environment into account, $\Ge$ is not guaranteed at all.
For the environment may fail to send the signal $\bar a$ that is received as $a$, in which case the
system will be stuck in its initial state.
\end{exam}
To properly model reactive systems, we distinguish \emph{blocking} and \emph{non-blocking} transitions.
A blocking transition---$a$ in our example---requires participation of the environment in which the system will be running,
whereas a non-blocking transition (e.g.\ $\tau$) does not. In \cite{Rei13}, blocking and non-blocking transitions
are called \emph{cold} and \emph{hot}. In many process algebras transitions are labelled with
actions, and whether a transition is blocking is entirely determined by its label.
Accordingly, we speak of blocking and non-blocking actions.
In CCS, the only non-blocking action is $\tau$.
However, for certain applications it makes sense to declare some other actions to be non-blocking
\cite{GH15b}; this constitutes a promise that we will never put the system in an environment that
can block these actions. In \cite{GH15a} we use a process algebra with broadcast communication:
a broadcast counts as a non-blocking action, because it will happen regardless whether anyone
receives it.

In the setting of reactive systems the progress assumption of \Sec{progress} needs to be
reformulated:
\begin{myquote}
A (transition) 
system in a state that admits a non-blocking transition will eventually progress, \ie perform a transition.
\end{myquote}

\begin{definition}
A path in a transition system is \emph{progressing} if either it is infinite or its last state is
the source of no non-blocking transition.
\end{definition}
Our justness assumption is adapted in the same vein:
\begin{myquote}%
  Once a non-blocking transition is enabled that stems from a set of parallel components,
  one (or more) of these components will eventually partake in a transition.
\end{myquote}

\begin{definition}
  A path $\pi$ in a transition system is \emph{just} if for each
  non-blocking transition $t$ with $s:=\source(t) \in \pi$, a transition $u$ occurs
  in $\pi$ past the occurrence of $s$, such that $t \naconc u$.
\end{definition}
The treatment of fairness is adapted to reactive systems by defining a task $T$ to be
\emph{enabled} in a state $s$ iff there exists a \emph{non-blocking} transition $t \in T$ with
$\source(t)=s$ (cf.\ \df{fair}). 
In \df{justice} $t$ is also required to be non-blocking.
The reason is that if sufficiently many states on a path merely
have an outgoing transition from $T$ that is blocking, it might very well be that the environment
blocks all those transitions, so that the system is unable to perform a transition from $T$. On the
other hand, performing a blocking transition from $T$ is a valid way to
satisfy the promise of fairness. To obtain a notion of fairness where an arbitrary non-blocking transition from
$T$ is required to occur in $\pi$, one can simply define a new task that only has the non-blocking
transitions from $T$.

With these amendments to the treatment of progress and fairness, the observation remains that
progress is exactly what weak or strong fairness prescribes for finite paths, provided that each
non-blocking transition occurs in a task.

Applied to CCS, where only $\tau$-actions are non-blocking, fairness of actions loses most of its
power. All other relations in the taxonomy remain unchanged. \thm{feasibility} remains valid as well.
\df{AGEF} needs to be reformulated as follows.

\begin{definition}
A liveness property $\Ge$, modelled as a set of states in a reactive transition system $G$, is an \emph{AGEF property}
iff (a state of) $\Ge$ is reachable along a non-blocking path\footnote{A path containing non-blocking transitions only.}
from every state $s$ that is reachable (by any means) from an initial state of $G$.
\end{definition}

With this adaptation, a non-AGEF property is never ensured, no matter which feasible fairness property is assumed.
Yet, \thm{ST-fairness} still holds, with the \emph{distance}, in the proof, being the length of the shortest non-blocking
path to $\Ge$.
We do not generalise probabilistic fairness to reactive systems.

Full fairness for reactive systems is what corresponds to \emph{Koomen's fair abstraction rule}
(KFAR), a widely-used proof principle in process algebra, introduced in \cite{BK86}; see also
\cite{BBK87a,BW90}. Any process-algebraic verification that shows a
liveness property of a system $P$ by proving $P$ weakly bisimilar to a system $Q$ where this
liveness property obviously holds, implicitly employs full fairness.

\section{Evaluating Notions of Fairness}\label{sec:eval}
The following table evaluates the notions of fairness reviewed in this paper against the criteria for
appraising fairness properties discussed in
Sections~\ref{sec:criteria},~\ref{sec:J-fairness},~\ref{sec:sw} and~\ref{sec:justness}.
\[\begin{array}{@{}l@{\;}|@{~}c@{~}c@{~}c@{~}c@{~}c@{~}c@{~}c@{~}c@{~}c@{~}c@{~}c@{~}c@{~}c@{~}c@{~}c@{~}c@{~}c@{~}c@{~}c@{~}c@{~}c@{}}
& \mbox{\footnotesize\Fu} & \mbox{\footnotesize\Ex} & \mbox{\footnotesize\St\Ac} & \mbox{\footnotesize\St\Ts} & \mbox{\footnotesize\St\Sy} & \mbox{\footnotesize\St\Gr} & \mbox{\footnotesize\St\In} & \mbox{\footnotesize\rm SWI} & \mbox{\footnotesize\St\Cp
           } & \mbox{\footnotesize\W\Ac} & \mbox{\footnotesize\W\Ts} & \mbox{\footnotesize\W\Sy} & \mbox{\footnotesize\W\Gr} & \mbox{\footnotesize\W\In} & \mbox{\footnotesize\W\Cp
           } & \mbox{\footnotesize\J\Ac} & \mbox{\footnotesize\J\Ts} & \mbox{\footnotesize\J\In} & \mbox{\footnotesize\J\Cp} & \mbox{\footnotesize\J} & \mbox{\footnotesize\Pr}\\
\hline
\textrm{feasibility}\rule{0pt}{11pt} & +&+&+&+&+&+&+&+&+&+&+&+&+&+&+&+&+&+&+&+&+
\\
\textrm{equivalence robustness} &      \mbox{\footnotesize(}\mathord+\mbox{\footnotesize)}&-&-&-&-&-&-&\pm&\pm&-&?&?&?&-&-&-&?&-&-&+&+
\\
\textrm{liveness enhancement} &        +&+&+&+&+&+&+&+&+&+&+&+&+&+&+&+&+&+&+&+&+
\\
\textrm{pres.\ under unfolding} &      +&+&+&-&+&+&+&+&+&+&-&+&+&+&+&+&-&+&+&+&+
\\
\textrm{pres.\ under action ref.} &    +&+&+&+&+&+&+&+&+&-&-&-&-&-&-&+&+&+&+&+&+
\\
\textrm{useful for queues} &           -&-&-&-&-&-&-&+&-&-&-&-&-&-&-&-&-&-&-&-&-
\\
\textrm{typically warranted} &         -&-&-&-&-&-&-&-&-&-&-&-&-&-&-&-&+&-&-&+&+
\end{array}\]
All notions reviewed satisfy the criterion of feasibility, from Section~\ref{sec:feasibility};
this is no wonder, since we disregarded any notion not satisfying this property.
The counterexamples given in Section~\ref{sec:robustness} against the equivalence robustness of
WC- and $S\Y$-fairness, for $y\in\{\text{A,T,Z,G,I}\}$,
can easily be adapted to cover the other notions adorned with a $-$ in the
table above; for SWI and SC this involves the use of $N$-way synchronisation, thereby leaving our
fragment of CCS\@. The notion is meaningless for full fairness (\Fu), and the
positive results for progress and justness are fairly straightforward.
It is trivial to find examples showing that JA and JT-fairness, and thus also all stronger notions,
are liveness enhancing w.r.t.\ the background assumption of `minimal progress' (cf.\ Section~\ref{sec:enhancement}).
As observed in Section~\ref{sec:unfolding}, only fairness of transitions fails to be preserved under
unfolding. Being not preserved under refinement of actions, as discussed in Section~\ref{sec:J-fairness},
is a problem that befalls only the notions of weak fairness. The next entry refers to the typical
application of fairness illustrated in \Sec{sw}; it is adequately handled only by
strong weak fairness. The last line reflects our observation that most notions of fairness are not
warranted in many realistic situations. This is illustrated by \ex{phone} in
Section~\ref{sec:justness} for weak and strong fairness, Example~\ref{ex:JA} for JA-fairness, and
Example~\ref{ex:JI} for JI- and JC-fairness.

\section{Conclusion, Related and Future work}\label{sec:conclusion}

\noindent
We compared and classified many notions of fairness found in the literature.
Our classification is by no means exhaustive. We skipped, for instance,
\emph{k-fairness} \cite{Be84} and \emph{hyperfairness} \cite{La00}---both laying
between strong fairness and full fairness. We skipped \emph{unconditional fairness}
\cite{Kw89}, introduced in \cite{LPS81} as \emph{impartiality}, since that notion does not satisfy
the essential criterion of feasibility. We characterised full fairness as the strongest possible
fairness assumption that satisfies this important requirement.

Comparisons of fairness notions appeared before in \cite{KdR83,AFK88,Jo01}.
{\sc Kuiper \& de Roever} \cite{KdR83} considered weak and strong fairness of components, instructions and
synchronisations, there called processes, guards and channels, respectively. 
Their findings, depicted in \fig{earlier}, agree mostly with ours, see \fig{taxonomy}. 
They differ in the inclusions between WZ, WI and WC, which are presented without proof and deemed ``easy''.
It appears that our Counterexamples~\ref{ex:WS} and~\ref{ex:WC} can be translated into the version
of CSP used in \cite{KdR83}.

\begin{figure}[ht]
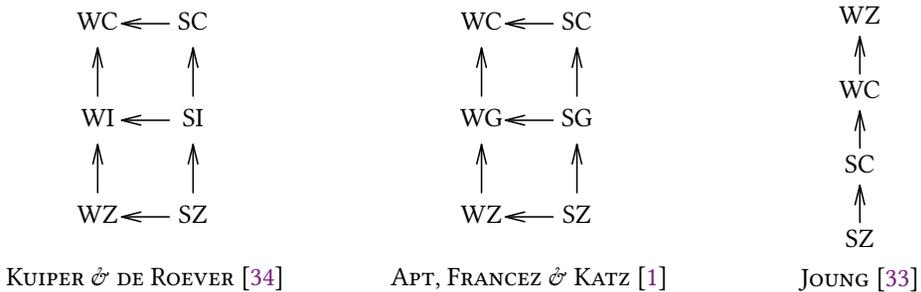

  \mbox{}\hfill
\input{latticeKdR}
{\raisebox{1ex}{\box\graph}}
  \hfill\hfill
\input{latticeAFK}
{\raisebox{1ex}{\box\graph}}
  \hfill\hfill
\input{latticeJOUNG}
{\raisebox{1ex}{\box\graph}}
  \hfill\mbox{}
\caption{Earlier classifications of fairness assumptions}\label{earlier}
\end{figure}
{\sc Apt, Francez \& Katz} \cite{AFK88} consider weak and strong fairness of components, groups of
components and synchronisations, there called processes, channels or groups, and communications, respectively.
The authors work with a restricted subset of CSP, on which the difference between
the weak  notions of fairness  and the corresponding notions of
J-fairness largely disappears.
Most of their strict inclusions (see \fig{earlier}) agree with ours,
except for the ones between WZ, WG and WC\@. Our counterexamples~\ref{ex:WS} and~\ref{ex:WC} cannot
be translated into their subset of CSP\@. In fact, it follows from 
\citep[Prop.\ 5]{AFK88}
that in their setting WZ, WG and WC are equally powerful.

{\sc Joung} \cite{Jo01} considers weak and strong fairness of components and synchronisations,
there called processes and interactions. He also considers many
variants of these notions, not covered here. Restricting his hierarchy to the four notions in
common with ours yields the third lattice of \fig{earlier}.
Joung makes restrictions on the class of modelled systems comparable with those of \cite{AFK88};
as a consequence weak fairness and J-fairness coincide. By this, his lattice
agrees with ours on the common elements.

In \Sec{fairness} we proposed a general template for strong and weak fairness assumptions, which
covers most of the notions reviewed here (but not probabilistic and full fairness).
We believe this template also covers interesting notions of fairness not reviewed here.
 In
\cite{CCP09}, for instance, two transitions belong to the same task in our sense iff they are both
$\tau$-transitions resulting from a $\pi$-calculus synchronisation of actions $\bar ay$ and $a(z)$
that both lay within the scope of the same restriction operator $(\nu a)$.
A more general definition of what counts as a fairness property, given in the form of a 
language-theoretic, a game-theoretic, and a topological characterisation, appears in~\cite{VV12}.

Using fairness assumptions in the verification of liveness properties amounts to saying that if one
tries something often enough, surely it will eventually succeed. Although for some applications this
is useful, in many cases it is unwarranted.
Progress assumptions, on the other hand, are warranted in many situations, and are moreover
indispensable in the verification of liveness properties.
Here we introduced the concept of \emph{justness}, a stronger version of progress, that we believe
is equally warranted, and, moreover, indispensable for proving many important liveness properties.
Special cases of justness are prevalent in the literature on fairness.
However, they often occur as underlying assumptions when studying fairness, rather than as 
alternatives. Moreover, we have not seen them occur in the general form we advocate here.

We showed that as a fairness assumption justness coincides with fairness of events.
This could be interpreted as saying that justness is not a new
concept. However, properly defining fairness of events requires a much more sophisticated machinery
than defining justness. More importantly, this machinery casts it as a fairness assumption, thus making it equally implausible
as other forms of fairness. By recognising the same concept as justness, the precondition of
an infinite series of attempts is removed, and the similarity with progress is stressed.
This makes justness an appealing notion.

As future work we plan to extend the presented framework---in particular justness---to formalisms 
providing some form of asymmetric communication such as a broadcast mechanism. Examples 
for such formalisms are AWN~\cite{FGHMPT12a} a (process-)algebra for wireless networks,
and Petri nets with read arcs~\cite{Vo02}.

\paragraph
{Acknowledgement:} We thank the referees for their careful reading and valuable feedback.

\bibliographystyle{ACM-Reference-Format}
\bibliography{justness}

\newpage
\appendix
\section{CCS}\label{app:CCS}

\begin{table*}[t]
\caption{Structural operational semantics of CCS}
\vspace{-2ex}
\label{tab:CCS}
\normalsize
\begin{center}
\framebox{$\begin{array}{c@{}c@{}c}
\alpha.E \goto{\alpha} E &
\displaystyle\frac{E\goto{\alpha} E'}{E+F \goto{\alpha} E'} &
\displaystyle\frac{F\goto{\alpha} F'}{E+F \goto{\alpha} F'} \\[4ex]
\displaystyle\frac{E\goto{\alpha} E'}{E|F \goto{\alpha} E'|F} &
\displaystyle\frac{E\goto{a} E' ,~ F \goto{\bar{a}} F'}{E|F \goto{\tau} E'| F'} &
\displaystyle\frac{F \goto{\alpha} F'}{E|F \goto{\alpha} E|F'}\\[4ex]
\displaystyle\frac{E \goto{\alpha} E'}{E\backslash a \goto{\alpha} E'\backslash a}~~(a{\neq} \alpha{\neq} \bar a) &
\displaystyle\frac{E \goto{\alpha} E'}{E[f] \goto{f(\alpha)} E'[f]} &
\displaystyle\frac{S(X)[\textbf{fix}_Y S/Y]_{Y\in \dom(S)} {\goto{\alpha}} E}{\textbf{fix}_XS\goto{\alpha}E}
\end{array}$}
\vspace{-2ex}
\end{center}
\end{table*}

CCS \cite{Mi80} is parametrised with a set ${\A}$ of {\em names}.
The set $\bar{\A}$ of {\em co-names} is $\bar\A:=\{\bar{a}
\mid a\mathbin\in {\A}\}$, and $\Lab:=\A \cup \bar\A$ is the set of \emph{labels}.
% The function $\bar{\cdot}$ is extended to $\Lab$ by declaring $\bar{\bar{\mbox{$a$}}}=a$.
Finally, \plat{$Act := \Lab\dcup \{\tau\}$} is the set of
{\em actions}. Below, $a$, $b$, $c$, \ldots range over $\Lab$ and
$\alpha$, $\beta$ over $Act$.
A \emph{relabelling} is a function $f\!:\Lab\mathbin\rightarrow \Lab$ satisfying
$f(\bar{a})=\overline{f(a)}$; it extends to $Act$ by $f(\tau):=\tau$.
Let $\V$ be a set $X$, $Y$, \ldots of \emph{process variables}.
The set $\IT_{\rm CCS}$ of CCS expressions is the smallest set including:
\begin{center}
\begin{tabular}{@{}lll@{}}
${\bf 0}$ && \emph{inaction}\\
$\alpha.E$  & for $\alpha\mathbin\in Act$ and $E\mathbin\in\IT_{\rm CCS}$ & \emph{action prefixing}\\
$E+F$ & for $E,F\mathbin\in\IT_{\rm CCS}$ & \emph{choice} \\
$E|F$ & for $E,F\mathbin\in\IT_{\rm CCS}$ & \emph{parallel composition}\\
$E\backslash a $ & for $a\in \A$ and $E\mathbin\in\IT_{\rm CCS}$ & \emph{restriction} \\
$E[f]$ & for $f$ a relabelling and $E\mathbin\in\IT_{\rm CCS}$ & \emph{relabelling} \\
$X$ & for $X\mathbin\in\V$ & \emph{process variable}\\
$\textbf{fix}_XS$ & for $S\!:\V\mathord\rightharpoonup \IT_{\rm CCS}$
and $X\mathord\in \dom(S)$ & \emph{recursion}.
\end{tabular}
\end{center}
A partial function $S\!:\V\mathbin\rightharpoonup \IT_{\rm CCS}$ is called a
\emph{recursive specification},\vspace{1pt}
and traditionally written as \plat{$\{Y\defis S(Y)\mid Y\mathbin\in \dom(S)\}$}.
We often abbreviate $\alpha.{\bf 0}$ by $\alpha$, and $\textbf{fix}_XS$ by ``$X$ where $S$''.
A CCS expression $P$ is \emph{closed} if each occurrence of a process variable $Y$ in $P$ lays within a
subexpression $\textbf{fix}_XS$ of $P$ with $Y\mathbin\in\dom(S)$; $\cT_{\rm CCS}$ denotes the set of closed CCS
expressions, or \emph{processes}.

The traditional semantics of CCS is given by the labelled transition relation
$\mathord\rightarrow \subseteq \cT_{\rm CCS}\times Act \times\cT_{\rm CCS}$ between
closed CCS expressions. The transitions \plat{$p\goto{\alpha}q$} with $p,q\mathbin\in\cT_{\rm CCS}$
and $\alpha\mathbin\in Act$ are derived from the rules of \tab{CCS}.

To extract from each CCS process $P$ a transition system $(S,\Tr,\source,\target,I)$ as in \df{TS},
augmented with a labelling $\ell:\Tr\rightarrow Act$,
take $S$ to be the set of closed CCS expressions reachable from $P$, $I=\{P\}$, and $\Tr$ the set
of \emph{proofs} $\pi$ of transitions \plat{$p\goto{\alpha}q$} with $p\in S$. Here a \emph{proof} of
\plat{$p\goto{\alpha}q$} is a well-founded tree with the nodes labelled by elements of
$\cT_{\rm CCS}\times Act \times\cT_{\rm CCS}$, such that the root has label \plat{$p\goto{\alpha}q$}, and if
$\mu$ is the label of a node and $K$ is the set of labels of the children of this node then
\plat{$\frac{K}{\mu}$} is an instance of a rule of \tab{CCS}.
Of course $\source(\pi)\mathbin=p$, $\target(\pi)\mathbin=q$ and $\ell(\pi)\mathbin=\alpha$.

\section{Defining instructions and components on a fragment of CCS}\label{app:fragment}

In \Sec{taxonomy} we restrict our attention to the fragment of CCS where $\dom(S)$ is finite for all
recursive specifications $S$ and parallel composition does not occur in the expressions $S(Y)$.
Moreover, each occurrence of a process variable $X$ as well as each
occurrence of a parallel composition $H_1|H_2$ in an expression $E+F$ is \emph{guarded}, meaning
that it lays in a subexpression of the form $\alpha.G$.
Given a process $P$, let $\I$ be the set of all occurrences of action prefix operators $\alpha.\_$ in $P$.
To define the function $\instr$, give each such action a different name $n$, and carry these names
through the proofs of transitions, by using\vspace{-1ex}
\[\alpha_n.E \goto{\alpha}_n E \qquad\qquad
  \displaystyle\frac{E\goto{a}_n E' ,~ F \goto{\bar{a}}_m F'}{E|F \goto{\tau}_{n,m} E'| F'} \]
and keeping the same list of names in each of the other rules of \tab{CCS}.
If $\pi\in\Tr$ is a proof of a transition \plat{$p\goto{\alpha}_Z q$} with $Z$ a list of names
(seen as a set), then $\instr(\pi)=Z$.
In \App{unique} we show that property \ref{unique synchronisation} holds for this fragment of CCS\@.
As CCS expressions are finite objects, property \ref{finite} holds as well.

The set $\Ce$ of parallel components of a closed expression $P$ in our fragment of CCS is defined as
the arguments ({\sc left} or {\sc right}) of parallel composition operators occurring in $P$ (or the
entire process $P$). A component
$C\in\Ce$ can be named by a prefix-closed subset of the regular language $\{\Left,\R\}^*$. 
Each action occurrence can be associated to exactly one of those components.
This yields a function $\cmp:\I\rightarrow\Ce$.

\begin{exam}
The CCS expression $a.(P|b.Q)|U$ has at least five components---more if $P$, $Q$ or $U$ have
parallel subcomponents---namely $a.(P|b.Q)$ (named $\Left$), $P$ (named $\Left\Left$), $b.Q$ (named $\Left\R$), and
$U$ (named $\R$), and the entire expression (named $\varepsilon$).
Let $a_1$ and $b_1$ be the indicated occurrences of $a$ and $b$. Then $\cmp(a_1)=\Left$ and $\cmp(b_1)=\Left\R$.
\end{exam}
Now the function $\comp\mathop{:}\Tr\mathbin\rightarrow\Pow(\Ce)$ is given by
$\comp(t)=\{\cmp(I)\mid I\mathbin\in \instr(t)\}$, so that \ref{cmp} is satisfied.

\section{Proof of the Unique Synchronisation Property}\label{app:unique}

An expression $F$ is called a \emph{subexpression} of a CCS expression $E$ iff $F \preceq E$,
where $\preceq$ is the smallest reflexive and transitive relation on CCS expressions satisfying
$E \preceq \alpha.E$,
$E \preceq E+F$, $F \preceq E+F$,
$E \preceq E|F$, $F \preceq E|F$,
$E \preceq E\backslash a$,
$E \preceq E[f]$ and
$S(Y) \preceq \textbf{fix}_X S$ for each $Y \in \dom(S)$.
An \emph{extended subexpression} is defined likewise, but with an extra clause
$\textbf{fix}_Y S \preceq \textbf{fix}_X S$ for each $X,Y\in \dom(S)$.
We say that $F$ has an \emph{unguarded occurrence} in $E$ iff $F \leq E$, where $\leq$ is defined
as $\preceq$, except that the clause $E \preceq \alpha.E$ is skipped, and the clauses for \textbf{fix} are
replaced by $S(X) \leq \textbf{fix}_X S$ and
\begin{center}
if $Y \leq \textbf{fix}_X S$ and $Y \in\dom(S)$ then $S(Y)\leq \textbf{fix}_X S$.
\end{center}
If $\alpha.F\leq E$ then that occurrence of $\alpha$ in $E$ is called \emph{unguarded}.

A \emph{named} CCS expression is an expression $E$ in our fragment of CCS in which each action
occurrence is equipped with a name. It is \emph{well-named} if for each extended
subexpression $F$ of $E$, all unguarded action occurrences in $F$ have a different name, 
and moreover, every subexpression $F|G$ of $E$
satisfies $n(F)\cap n(G)=\emptyset$, where $n(E)$ is the set containing all names of action occurrences in $E$.
Clearly, if all its action occurrences have different names, $E$ is well-named.
\vspace{1ex}

\textsc{Claim 1:}
If $E$ is well-named, then so is any subexpression $F$ of $E$, and $n(F)\subseteq n(E)$.
\vspace{1ex}

\textsc{Proof of Claim 1:} Directly from the definitions.
\hfill \rule{1ex}{7pt}\vspace{1ex}

\textsc{Claim 2:} If $E:=\textbf{fix}_{X} S$ is well-named
then so is $H:=S(X)[\textbf{fix}_Y S/Y]_{Y\in \dom(S)}$, and $n(H)\subseteq n(E)$.%
\vspace{1ex}

\textsc{Proof of  Claim 2:}
Clearly $n(H) \subseteq n(S(X)) \cup \bigcup_{Y\in \dom(S)} n(\textbf{fix}_Y S) \subseteq n(E)$.

By the restrictions imposed on our fragment of CCS, $H$ has no subexpressions of the
form $F|G$. Let $F$ be an extended subexpression of $H$. We have to show that all unguarded action
occurrences in $F$ have different names. Since $F$ is an extended subexpression of $H$, it either
is an extended subexpression of $\textbf{fix}_Y S$ for some $Y\in \dom(S)$, or of the form
$F'[\textbf{fix}_Y S/Y]_{Y\in \dom(S)}$ for an extended subexpression $F'$ of $S(X)$.

In the first case $F$ is also an extended subexpression of $\textbf{fix}_X S$. Since the latter expression is
well-named, all unguarded action occurrences in $F$ have a different name.

In the second case we consider two subcases.
First suppose that no variable $Y\in\dom(S)$ has an unguarded occurrence in $F'$.
Then all unguarded action occurrences in $F$ are in fact unguarded action occurrences in $F'$.
Since $F'$ is an extended subexpression of $E$, they all have a different name.

Now suppose that a variable $Y\in\dom(S)$ has an unguarded occurrence in $F'$.
Then, by the restrictions imposed on our fragment of CCS, $F'$ has no occurrences of either choice,
parallel composition or action prefixing, and $Y$ is the only process variable occurring unguarded
in $F'$. So the unguarded action occurrences in $F$ are in fact unguarded action occurrences in
$\textbf{fix}_Y S$. Since the latter is an extended subexpression of $E$, they all have different names.
\hfill \rule{1ex}{7pt}\vspace{1ex}

\textsc{Claim 3:} If $E$ is well-named and \plat{$E\goto{\alpha}_Z F$}, then so is $F$.
\vspace{1ex}

\textsc{Proof of Claim 3:}
Using and Claims 1 and 2, a trivial structural induction on the proof of transitions shows that if $E$ is
well-named and \plat{$E\goto{\alpha}_Z F$}, then so is $F$, and $n(F)\subseteq n(E)$.
\hfill \rule{1ex}{7pt}\vspace{1ex}

We write $\pi:E {\rightarrow_n}$ if $\pi$ is a proof of $E \goto{\alpha}_Z F$ for some $\alpha$, $F$
and $Z$ such that $n\in Z$.
\vspace{1ex}

\textsc{Claim 4:} If $\pi:E{\rightarrow_n}$ then $E$ has an unguarded action occurrence named $n$.
\vspace{1ex}

\textsc{Proof of Claim 4:}
A straightforward induction on $\pi$.
\hfill \rule{1ex}{7pt}\vspace{1ex}

\textsc{Claim 5:}
If $E$ is well-named and $Z \subseteq \I$ then 
there is at most one proof of a transition \plat{$E \goto{\alpha}_Z F$}.
\vspace{1ex}

\textsc{Proof of Claim 5:}
A straightforward induction on the proof of \plat{$E\hspace{-1.1pt} \mathop{\goto{\alpha}_Z}\hspace{-1.1pt} F\!$}, using Claims 1,\! 2 and 4.%
\hfill \rule{1ex}{7pt}\vspace{1ex}

That property \ref{unique synchronisation} holds for the fragment of CCS of \App{fragment} follows from Claims 3 and 5.

\end{document}